\documentclass[12pt]{article}

\newcommand{\dtcolornote}[3][]{}

\usepackage{booktabs}
\usepackage{threeparttable}
\usepackage{siunitx}
\usepackage{graphicx}
\usepackage{subcaption}
\usepackage{amsmath}
\usepackage{natbib}
\let\cite\citeyear

\usepackage{setspace}
\usepackage{tikz}\usetikzlibrary{arrows.meta, shapes.geometric, positioning}
\usetikzlibrary{shapes.geometric}\tikzset{
  main node/.style={
    ellipse,
    draw,
    font=\sffamily\small\bfseries
  }
}

\usepackage{mathtools}
\usepackage{algorithm}
\usepackage{algpseudocode}
\usepackage{indentfirst}
\usepackage{pythonhighlight}
\usepackage{multirow}
\usepackage[T1]{fontenc}
\usepackage{amsmath,amssymb}
\usepackage{float}
\usepackage{hyperref}
\usepackage{xurl}
\usepackage{xcolor}

\usepackage{amsfonts}
\usepackage[utf8]{inputenc}
\usepackage{amsthm}

\usepackage{longtable}
\usepackage{booktabs, caption, makecell}
\usepackage{threeparttable}
\usepackage{graphicx}
\usepackage{parskip}
\usepackage{adjustbox}
\usepackage{tabularx}
\usepackage{listings}
\usepackage{multirow}
\usepackage{titlesec}

\definecolor{rune}{HTML}{4A6672}

\theoremstyle{plain}
\newtheorem{proposition}{Proposition}[section]

\newtheorem{theorem}{Theorem}[section]

\newtheorem{lemma}[theorem]{Lemma}
\DeclareMathAlphabet{\mathpzc}{OT1}{pzc}{f}{it}

\definecolor{darkgreen}{rgb}{0,0.6,0}

\usepackage[top=1.25in, bottom=1.25in, left=1.25in, right=1.25in]{geometry}

\title{A Model and Estimation of the Bitcoin Transaction Fee}

\author{Daniel J. Aronoff\thanks{MIT Department of Economics}, Kristian Praizner\thanks{MIT EECS}, Armin Sabouri\thanks{MIT Media Lab Digital Currency Initiative}}

\begin{document}
\maketitle

\pagenumbering{roman}
\begin{abstract}

Bitcoin transaction fees will become more important as the block subsidy declines, but fee formation is hard to study with blockchain data alone because the relevant queueing environment is unobserved. We develop and estimate a structural model of Bitcoin fee choice that treats the mempool as a market for scarce blockspace. We assemble a novel, high-frequency mempool panel, from a self-run Bitcoin node that records transaction arrivals, exits, block inclusion, fee-bumping events, and congestion snapshots. We characterize the fee market as a Vickery-Clarke-Groves mechanism and derive an equation to estimate fees. In the first-stage we estimate a monotone delay technology linking fee-rate priority and network state to expected confirmation delay. We then estimate how fees respond to that delay technology and to transaction characteristics. We find that congestion is the main determinant of delay; that the marginal value of priority is priced in fees, which is increasing in the gradient of confirmation time reduction per movement up in the fee queue; and that transactor choice of RBF, CPFP, and block conditions have economically important effects on fees.


\end{abstract}

\noindent\textit{ACM CCS:} Applied computing $\rightarrow$ Electronic commerce; Theory of computation $\rightarrow$ Algorithmic mechanism design; Computing methodologies $\rightarrow$ Machine learning; Security and privacy $\rightarrow$ Distributed systems security; Mathematics of computing $\rightarrow$ Regression analysis.

\pagebreak

\pagebreak

\newpage
\pagenumbering{arabic}
\onehalfspacing

\section{Introduction}

Bitcoin is built on a proof-of-work protocol that requires miners to assemble transactions into blocks and compete to solve a puzzle in a lottery where the winner appends its block to the chain and earns a block reward. The block reward is the mining revenue which provides incentive for a miner to invest in hashrate (puzzle guesses per unit time) to solve the puzzle. Hashrate consumes electricity, which has a private cost to the miner.
The block reward serves two functions. One function is to ensure the liveness of the blockchain - i.e. that miners append transactions to each block and that a new block is added approximately every 10 minutes - by ensuring that computers are engaged in mining at all times. The other function is to enhance network security by erecting a threshold expected cost that an attacker, who builds a fork chain, must exceed to make its fork the consensus chain.
The equilibrium condition for Bitcoin mining implies that hashrate increases and decreases with the value (usually expressed in USD) of the block reward 
Figure \ref{fig:electricity_rev} displays the empirical relationship between the block reward (in USD) and energy consumption over time. The monotonicity of the relationship validates the prediction of the theoretical model.

\begin{figure}[H]
\begin{center}
\includegraphics[page=1,width=0.9\textwidth,height = 0.41\textheight]{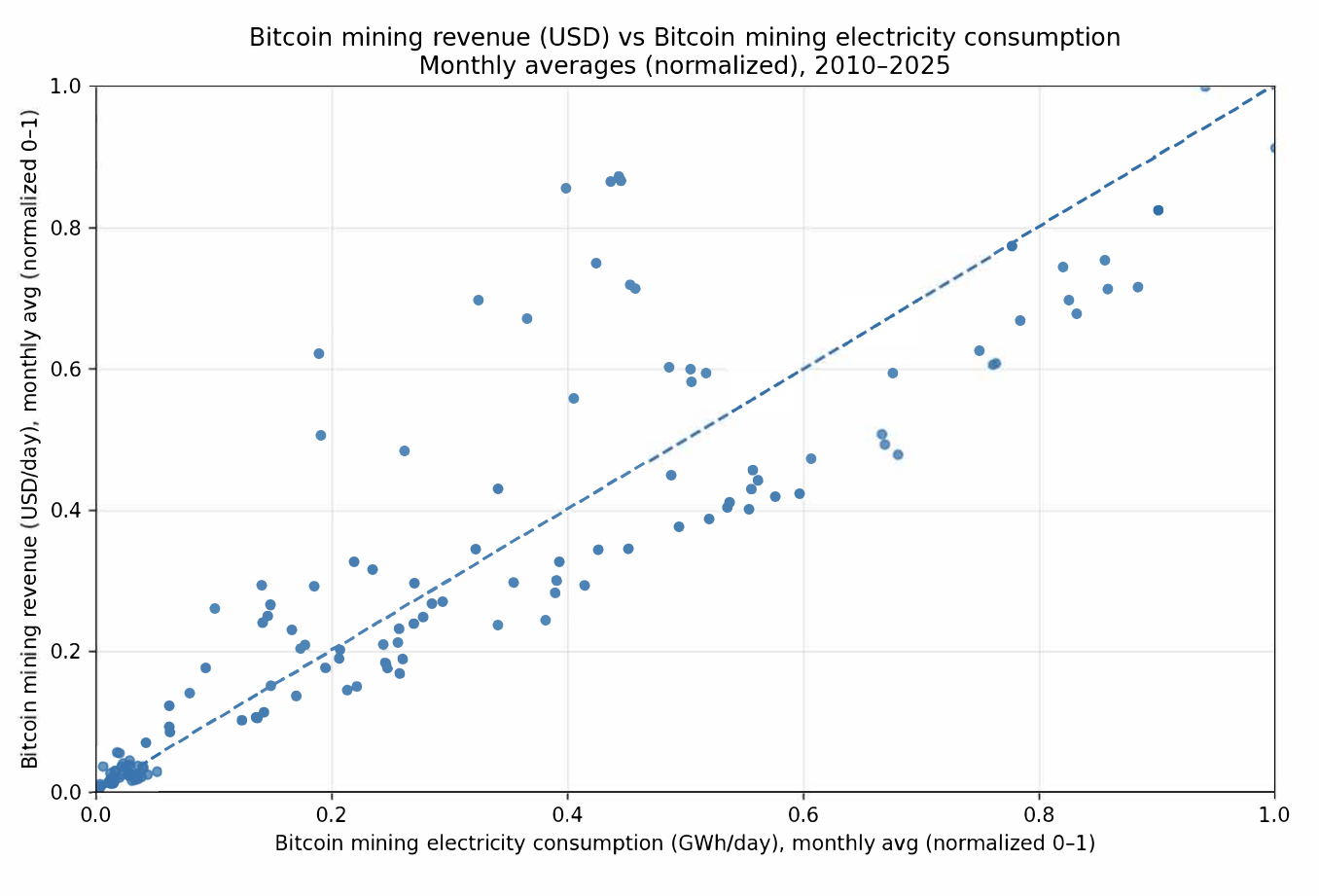}
\end{center}
\caption{Source: Cambridge Centre for Alternative Energy}
\label{fig:electricity_rev}
\end{figure}

The block reward has two parts: an algorithmic subsidy and voluntary transaction fees. The subsidy is programmed to fall through periodic halvings, so a larger share of miner revenue will eventually have to come from fees. A transactor’s motivation to pay a fee is to secure a priority position in the transaction queue; a higher fee induces the miner to append the transaction to its block ahead of a lower fee transaction, thereby shortening the expected confirmation time. A useful model of the fee market will recover two distinct objects: the mapping from priority into delay, and the user-side willingness to pay for a shorter wait.

We study that problem with a novel transaction-level mempool data-set collected from a self-run Bitcoin Core node operated by the  Medai Lab Digital Currency Initiative. The data record transaction entry into and exit from the mempool, block inclusion, replace-by-fee events, child-pays-for-parent relationships, and high-frequency snapshots of congestion. Those data enable us to reconstruct the queueing environment each transaction faced when its fee was chosen, which is the crucial object missing from standard blockchain-only data. Our empirical strategy follows the economics closely. In the first stage, we estimate expected confirmation delay as a flexible function of fee-rate priority and network conditions, impose the economically sensible restriction that higher priority cannot raise expected delay, and recover the local slope of the delay schedule. In the second stage, we relate fees to that estimated slope and to transaction characteristics in a log-fee equation with epoch fixed effects. We also examine a re-spend-based impatience proxy in supplementary specifications, but the baseline estimates do not depend on it.

The main findings are straightforward. Congestion is the dominant determinant of confirmation delay, and the delay schedule  -- i.e. expected delay between size-ranked fees -- becomes steeper when the mempool is congested Fees move with that technology in the direction implied by the model: when a marginal increase in priority buys a larger reduction in delay, users pay more for it. Several institutional details matter as well. RBF-tagged transactions pay large premia, CPFP packages are priced differently from stand-alone transactions, and fees rise with blockspace utilization, mempool depth, and the time since the previous block. The model is most informative within epochs, where it explains a substantial share of cross-sectional fee variation; epoch-to-epoch shifts in the overall fee level are harder to account for and are largely absorbed by fixed effects. 

One surprising result is that  the number of blocks to re-spend UTXO's is not informative. M\"oser and B\"ohme \cite{Bohme2015} found this measure to be correlated with the presence of a transaction fee on pre-2015 Bitcoin data. Intuitively, a fast turnaround by a receiver of UTXO's should be a motive to pay a fee to ensure the transaction is processed quickly. We expected to find this correlation to hold in our data, but it does not. 

This paper sits between several literatures. \citep{Huberman2021} provide the closest theoretical starting point by modeling blockchain fees as payments for priority in a queue. \citep{Easley2019} study fee payment using lower-frequency reduced-form data, while \citep{lehar2020miner} use transaction-level data to study miner market power. Our contribution is to take a structural priority-pricing framework to transaction-level mempool data.

The mempool data spans only a little more than nine weeks, so long-run extrapolation requires caution, particularly because the sample comes from a period of relatively stable network operation, there is limited variation in the explanatory variables of our model. We validate this by comparing the variance of key on-chain variables in our sample to the prior three years. Our remedy is to continue running the DCI Bitcoin node to lengthen our sample over time.

The remainder of the paper proceeds as follows. Section \ref{sec:background} describes Bitcoin fee formation and the mempool. Section \ref{sec:model} develops the model. Section \ref{sec:data_measurement} describes the data. Section \ref{sec:estimation} presents the empirical strategy. Section \ref{sec:results} reports the results. Section \ref{sec:conclusion} concludes.

\section{Background}
\label{sec:background}

This section summarizes the institutional features of Bitcoin that govern fee formation. We clarify what economic object a transaction fee represents, how miners select transactions into blocks, and why confirmation delay is a function of transaction priority and network congestion. These mechanics motivate the sparse model in Section~\ref{sec:model} and the empirical design in Section~\ref{sec:estimation}.

\subsection{Transactions, blockspace, and the block reward}
\label{subsec:bg_blockreward}

A Bitcoin payment is implemented by a transaction that spends existing unspent outputs (UTXOs) and creates new outputs.\footnote{An output specifies a spending condition (a script) and an amount. Outputs that have not been spent are UTXOs.} Transactions are broadcast to the peer-to-peer network and, if valid, are stored by nodes as unconfirmed transactions.

Miners assemble valid unconfirmed transactions into a candidate block and compete to solve a proof-of-work puzzle. The first miner to find a valid solution propagates its block; other nodes accept it if it satisfies consensus rules. The winning miner earns the block reward, which has two components: (i) an algorithmically determined subsidy (paid in a special ``coinbase'' transaction)
and (ii) the sum of transaction fees from the transactions included in the block.

Blockspace is scarce. Bitcoin constrains block capacity using a weight-based rule, so transactions differ in the amount of capacity they consume. This institutional detail makes fee comparisons across transactions naturally scale by size: a miner deciding which transactions to include ranks profitability by comparing revenue per unit of blockspace rather than revenue per transaction.\footnote{Strategic behavir by miners can, in certain circumstances, cause them to choose a lower profit transaction \citep{lehar2020miner}.}

\subsection{What a transaction fee is and why fee rate is the relevant price}
\label{subsec:bg_fee}

A transaction fee is voluntary and is determined by the sender. Mechanically, it equals the difference between the total value of inputs and the total value of outputs.\footnote{The fee is not an explicit field in the transaction format. A miner can infer it from inputs and outputs.} Because transactions differ in how much blockspace they consume, the relevant price for priority is the fee per unit of size. We therefore work with the fee rate

\[ r \equiv \frac{b}{\mathrm{Weight}},\]

where $b$ is the fee and $\mathrm{Weight}$ is the transaction's blockspace footprint. In practice, both wallets and miners quote fee rates (often in satoshis per unit of weight) rather than fee levels.\footnote{See Section \ref{sec:data_measurement} for the technical definition of weight.}

The economic interpretation is immediate: blockspace is the scarce resource, and the fee rate is the marginal willingness to pay for that resource. Holding the fee rate fixed, a larger transaction pays a higher fee because it uses more blockspace. This scaling is the reason transaction weight enters multiplicatively in the model in Section~\ref{sec:model} and as a control in the estimation in
Section~\ref{sec:estimation}.

\subsection{The mempool as a priority queue}
\label{subsec:bg_mempool}

Nodes maintain a mempool: a local inventory of valid, unconfirmed transactions. The mempool is not a centralized object; it is replicated across nodes and can differ slightly across them because of propagation delays and node-specific policies. Nevertheless, for liquid periods with broad connectivity, miners and users face a common environment in which pending transactions compete for inclusion in future blocks.

From the standpoint of fee formation, the mempool behaves like a priority queue with random service times. The service opportunities are block arrivals, which occur stochastically, and the service capacity of each block is limited by the block weight constraint. When congestion is low, most valid transactions are confirmed quickly and the marginal value of priority is small. When congestion is high, expected delay is sensitive to priority and small fee-rate differences can
meaningfully affect confirmation outcomes.

These mechanics motivate two empirical objects that appear throughout the paper.
First, a transaction's priority can be represented by its fee-rate rank (or percentile) among pending or newly arriving transactions in a short interval. Second, confirmation delay can be summarized by an expected delay function that depends on both priority and the contemporaneous state of congestion (backlog) and capacity utilization.

\subsection{Block assembly and fee-based selection}
\label{subsec:bg_blockassembly}

Given a set of unconfirmed transactions, a miner chooses which transactions to include subject to the block weight limit. Under competitive mining and absent strategic withholding, the dominant selection rule is fee-based: miners prefer transactions that deliver higher expected fee revenue per unit of blockspace. In the baseline case of our theoretical model (Section \ref{sec:model} each transaction is independent with identical size (in terms of weight, not fee). This implies sorting by fee rate and filling block capacity from the top of the list.
Two practical features complicate the one-transaction-at-a-time description but preserve the fee-rate logic at the relevant margin.

\textbf{First}, transactions can depend on one another. A child transaction that spends an unconfirmed parent output is not valid for inclusion unless the parent is included. As a result, miners evaluate packages of related transactions using an effective fee rate based on combined fees and combined weight.

\textbf{Second}, senders can adjust fees after broadcast using widely used fee-bumping mechanisms. Replace-by-fee (RBF) allows a sender to replace an unconfirmed transaction with a higher-fee version. Child-pays-for-parent (CPFP) allows a sender or receiver to attach a high fee to a child transaction, inducing miners to include both parent and child because the package fee rate is attractive. These mechanisms sharpen the economic interpretation of observed fees: a fee is a choice variable used to purchase priority in a congested queue, potentially updated in response to
realized waiting.

These considerations motivate the additional indicators used in the empirical specification in Section~\ref{sec:estimation}. In particular, $\mathrm{RBF}$ and $\mathrm{CPFP}$ capture fee-bumping and package-selection behavior that affects the mapping between an individual transaction's posted fee rate and its effective priority.

\subsection{Implications for modeling and estimation}
\label{subsec:bg_implications}

The institutional details above imply three restrictions that guide the remainder of the paper.

\begin{enumerate}
\item the relevant object for priority is the fee rate, not the fee level, because the capacity constraint is weight-based. This is why Section~\ref{sec:model} indexes priority by fee-rate rank and why Section~\ref{sec:estimation} controls for transaction weight.

\item confirmation delay is jointly determined by priority and the state of the mempool. In the model in Section~\ref{sec:model}, we write expected delay as $W(p,s)$, where $p$ is a priority percentile and $s$ summarizes congestion and capacity conditions. In the estimation in Section~\ref{sec:estimation}, we approximate this mapping using flexible predictors of congestion and unused blockspace, together with transaction-level indicators that affect the mempool wait time of the transaction.

\item the economic force behind fee formation is the tradeoff between fees and delay. Wallets and transactors choose fees to target confirmation within a desired horizon, and miners select transactions using fee-based priority. The resulting fee schedule is therefore steep in periods of tight blockspace and flat in periods of slack blockspace. This is the core mechanism captured by the sparse model in Section~\ref{sec:model} and operationalized by the two-stage
estimator in Section~\ref{sec:estimation}.
\end{enumerate}



\section{The Transaction Fee Model}
\label{sec:model}

This section adapts the fee-setting framework in Huberman et.al. \cite{Huberman2021} to derive a parsimonious, estimable relationship between transaction fees, congestion, and  impatience (equivalently, delay cost). The economic object of interest is the fee a transactor attaches to a transaction to
obtain priority in the mempool and reduce confirmation delay.

\subsection{Environment and user problem}

Time is continuous. Transactions arrive to the mempool and are confirmed when included in a block. Blocks arrive as a Poisson process with rate $\mu$. Each block has a fixed capacity. In Huberman et.al. \cite{Huberman2021} the capacity is a fixed number of transactions and all transactions are the same size, so priority can be indexed by
the fee $b$. In the Bitcoin protocol, block capacity is a weight constraint and transactions differ in  weight; miners therefore prioritize by fee rate. Throughout we index priority by the fee rate
$r_{it} \equiv b_{it}/\mathrm{Weight}_{it}$, where $i$ indexes the transactor and $t$ indexed time.\footnote{The transaction fee is part of the bargain between counterparties, who share the surplus. it is paid by the sender, whom we refer to as the ``transactor''.}

A transactor $i$ has a willingness-to-pay $R_{it}$ to use Bitcoin relative to its best alternative, and a per-block delay cost $c_{it}$.
If the transaction is confirmed after expected delay $W_{it}$ (in blocks) and pays fee $b_{it}$, expected surplus is

\begin{equation}
u_{it} \;=\; R_{it} - c_{it} W_{it} - b_{it}.
\label{eq:util}
\end{equation}

The willingness-to-pay term aggregates all non-delay attributes of the payment system (including, in particular, the value of settlement finality and perceived reliability) relative to the outside option.
Conditional on choosing Bitcoin, transactor $i$ chooses a fee rate $r_{it}\ge 0$ (equivalently, a fee
$b_{it}=r_{it}\mathrm{Weight}_{it}$) to maximize \eqref{eq:util}.\footnote{Subject to $u_{it} \geq 0$.}

\subsection{Priority rule and delay as a function of rank}

A key implication of the Huberman et.al. \cite{Huberman2021} model is that, under competitive mining and free entry, miners adopt the myopic revenue-maximizing selection rule: each mined block includes the highest-paying pending transactions up to capacity. In their equal-size benchmark, this rule is “highest fees”; in our setting it is “highest fee rates.” Let $G_t(\cdot)$ denote the cross-sectional distribution of fee rates among relevant mempool transactions in epoch $t$. Define the priority percentile

\begin{equation}
p_{it} \equiv \frac{\mathrm{rank}_t(r_{it})}{N_t}\in(0,1),
\label{eq:priority_def_s3}
\end{equation}

where $\mathrm{rank}_t(\cdot)$ ranks fee rates in ascending order so that higher $p_{it}$ means higher priority, and $N_{t}$ is the number of transactions in the mempool at time $t$. Under the priority rule, expected delay is decreasing in priority. We therefore write expected delay as a function of priority and the contemporaneous state of the
mempool,

\begin{equation}
W_{it} \;=\; \mathcal{W}(p_{it}, s_t),
\label{eq:delay_schedule}
\end{equation}

where $s_t$ summarizes congestion and capacity conditions (e.g., backlog, blockspace utilization,
and other predetermined state variables).

\subsection{Equilibrium pricing of priority and a useful decomposition}

The equilibrium has two properties that are central for our empirical strategy.

\textbf{First}, fees implement an efficient priority allocation: higher delay-cost transactors obtain higher priority and shorter expected delays.

\textbf{Second}, the equilibrium payments coincide with the Vickery-Groves-Clarke (``VCG'') payments for selling priority of service, so each transactor pays the expected delay externality imposed on lower-priority transactions.

An immediate implication is that, provided willingness-to-pay is sufficiently high so that participation constraints do not bind, equilibrium fees are independent of $R_{it}$: changes in
outside options shift surplus but do not shift the priority-pricing schedule. This motivates our focus on fee choice conditional on observing that the transaction is submitted and confirmed.

For estimation, the VCG interpretation yields a convenient reduced-form decomposition.\footnote{See Appendix \ref{app:vcg_indirect} for a proof and an explanation of why VCG is an appropriate characterization of the Bitcoin transaction fee market.} Locally, raising priority by a small amount reduces expected delay by approximately $\mathcal{W}_p(p_{it}, s_t)\,dp$. Since the transactors value a reduction in delay, we write  $D(p_{it},s_t) = - \mathcal{W}_p =  - \partial \mathcal{W}(p_{it}, s_t)/\partial p < 0$, where $  \mathcal{W}_{p}$ denotes the marginal reduction in expected confirmation delay generated by an increase in priority, i.e. the marginal reduction in confirmation delay achieved by moving up one priority percentile, and $D(p_{it},s_t)$ is the absolute value of the local slope of the delay function. The marginal willingness to pay for priority is therefore proportional to the transactor’s delay cost $c_{it}$ times the reduction in confirmation time achieved by increasing priority by one position, $D(p_{it},s_t)$. This motivates the approximation

\begin{equation}
b_{it}\;\approx\;\kappa_t\cdot c_{it}\cdot D(p_{it},s_{t})\cdot \mathrm{Weight}_{it}
\label{eq:local_pricing}
\end{equation}

where $\kappa_t$ is a normalization constant that depends on the priority index and capacity constraint. Log linearization gives

\begin{equation}
\log b_{it}=\underbrace{\log c_{it}}_{\text{preferences}}+\underbrace{\log D(p_{it}, s_t)}_{\text{delay technology}}+\log \mathrm{Weight}_{it}+\text{const.}+\varepsilon_{it}.
\label{eq:log_decomp}
\end{equation}

Equation \eqref{eq:log_decomp} is the bridge to our empirical design: we require (i) a flexible estimate of the delay technology and its local slope in the prevailing state, and (ii) a flexible mapping from observables to the latent delay cost $c_{it}$. Section~\ref{sec:estimation} implements this logic using a two-stage estimator.


\section{Data and Measurement}
\label{sec:data_measurement}

\subsection{Data Sources}
\label{subsec:data}

Our empirical analysis relies on high--frequency mempool data collected from a self--run Bitcoin Core (version~30.0.0) node operating on Ubuntu~24.04.3~LTS.  The node runs with \verb|txindex=1|, RPC enabled and ZMQ enabled so that most valid transactions broadcasted to the peer--to-peer network are observed.

A separate Rust program, \emph{mempool--monitor} \footnote{https://github.com/arminsabouri/Mempool-Monitor/}, subscribes to the node's ZMQ \verb|rawtx| feed and receives a byte stream for each new transaction.  After decoding the raw transaction and pruning witness data, the monitor writes the raw transaction, the time of entry and exit and other metadata to a SQLite database.

In parallel, the monitor records the state of the node's mempool every $25$~seconds.  These snapshots capture the total mempool size (bytes), the number of pending transactions, and the current block height and hash.  The monitor also records \emph{Replace--By-Fee (RBF) events}: if a transaction that was previously seen is replaced with a higher--fee version, the new transaction is marked as an RBF and linked to the transaction it replaces. The Mempool--monitor  additionally identifies \emph{child--pays-for--parent (CPFP)} relationships by tracking chains of unconfirmed transactions; this information allows fee rates to be aggregated at the package level when parents and children are mined together. We track the time to re-spend the output as the number blocks between confirmation of the subject transaction and confirmation of transactions that use the output. Finally, we post-process and normalize the data before training. \footnote{https://github.com/mit-dci/bitcoin-fee-estimation}

Data collection began on \textbf{3~August~2025 \ at 06:27:44} (UTC$-04$) and ended on \textbf{9~October~2025 \ at 20:08:49} (UTC$-04$). Our raw data can be found on Kaggle. \footnote{https://www.kaggle.com/datasets/kristianpraizner/mempool-space-data/data}

\subsection{Preliminary Observations}

Figure 2 is a chart of the fee distribution of in our data, sorted by fee order. Upon visual inspection it has a convex shape that tends to increase in the order (from left to right). This motivates our choice of a log form of the VCG equation (Equation \ref{eq:log_decomp}) as the basis of our estimating equation.

\begin{figure}[h]
  \centering
  \includegraphics[width=\linewidth]{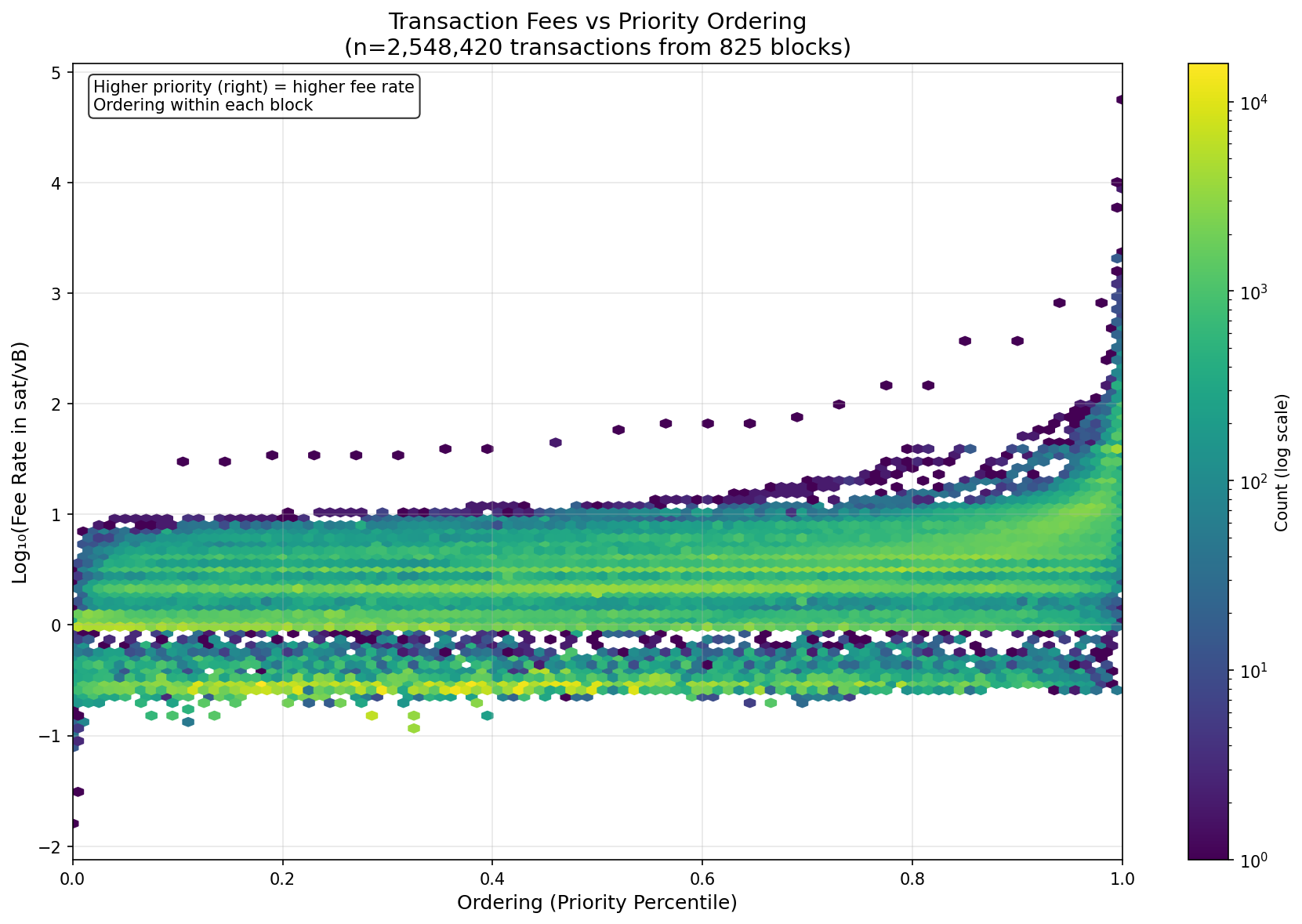}
  \caption{Log$_{10}$ fee rate (sat/vB) plotted against priority percentile for 2,548,420 confirmed transactions drawn from 825 blocks. Each point represents a transaction, colored by count density (log scale). Higher priority percentile (right) corresponds to higher fee rate, as transactions are ordered within each block by descending fee rate. The convex, upward-sloping relationship motivates the log-linear specification in Equation 5.}
  \label{fig:log_fee_vs_ordering}
\end{figure}


\section{Empirical Strategy and Estimation}
\label{sec:estimation}

This section estimates the fee-setting relationships implied by the model in
Section~\ref{sec:model}, specifically Equation \ref{eq:log_decomp}. The theoretical framework is intentionally sparse: it treats expected confirmation delay as a function of transaction priority and the mempool state, and it relates equilibrium fees to users' delay costs and to the local slope of the delay schedule. In the data, however, transaction fees also vary systematically with observables such as transaction weight, wallet and transaction type, and high-frequency network conditions. Our empirical strategy therefore augments the model with these controls while preserving its core decomposition between (i) delay technology and (ii) user preferences.

We proceed in two stages. In the first stage, we estimate expected confirmation
delay as a flexible function of priority and network state using a random forest,
and we recover an observation-level estimate of the local slope of this delay
schedule. In the second stage, we estimate a log fee equation that flexibly
maps an impatience proxy into fees using a cubic B-spline with four knots
(degree $3$), while controlling for the estimated local slope and additional
fee shifters. We implemented with Python code from the scikit-learn library. \footnote{
\label{fn:sklearn}
\textbf{scikit-learn implementation.}
The first stage uses \path{sklearn.ensemble.RandomForestRegressor}. The second stage uses a spline-basis expansion via \path{sklearn.preprocessing.SplineTransformer} with \path{degree=3} and \path{n_knots=4}, followed by OLS via \path{sklearn.linear_model.LinearRegression}, typically combined as a pipeline using \path{sklearn.pipeline.make_pipeline}. The scikit-learn documentation pages for these classes include a \texttt{[source]} link that points to the corresponding source code in the scikit-learn repository.}

\subsection{Data structure and notation}
\label{subsec:data_notation}

We index transactions by $i$ and discrete time intervals by $t$.
A time interval is an epoch (30 minutes in the baseline), chosen to balance two considerations: it is short enough that the state of the mempool is approximately stable within an interval and long enough to deliver sufficient within-epoch variation in priority.

Let $b_{it}$ denote the transaction fee paid by transaction $i$ (in USD) and let $\mathrm{Weight}_{it}$ denote its transaction weight.
Define the fee-rate

\begin{equation}
r_{it} \equiv \frac{b_{it}}{\mathrm{Weight}_{it}}.
\label{eq:feerate_def}
\end{equation}

Let $W_{it}$ denote the realized confirmation delay, measured as the number of blocks between mempool entry and inclusion. Let $\hat\rho_t$ denote a measure of congestion in epoch $t$ (e.g., the amount of mempool backlog in weight units) and let $\mathrm{Blockspace}_t$ denote available block capacity during the epoch. Let $\mathrm{Type}_{it}$ denote transaction-type indicators (including exchange-flow tags and NFT tags), and let $\mathrm{Wallet}_{it}$ denote wallet indicators when attribution is available.

Transactor impatience is not observed directly. Following the measurement strategy in Section~\ref{sec:data_measurement}, we use re-spend timing to proxy impatience. Let $d_{it}$ be the average number of blocks until the first $n$ output re-spends associated with transaction $i$, truncated at $X$ blocks as described above. Define the impatience proxy

\begin{equation}
\iota_{it} \equiv d_{it}^{-1}.
\label{eq:impatience_def}
\end{equation}

\paragraph{Variables.} We partition the variables into ``transaction variables'' $X_{it}$, and  ``state variables'' $S_{t}$, where $i$ indicates a transaction and $t$ is a (time interval) epoch.

$X_{it} =  \text{\Large}\{\mathrm{Blockspace}_t,\log \mathrm{Weight}_{it},
\mathrm{RBF}_{it},\mathrm{CPFP}_{it},\mathrm{Type}_{it},\mathrm{TOA}_{it}\text{\Large}\}$

$S_{t} = \text{\Large}\{\widehat\rho_t, \mathrm{Blockspace}_t\text{\Large}\}$

The concatenated bundle is denoted $Z_{it} \equiv \text{\Large}(X_{it}, S_{t}\text{\Large})$.

\subsection{Choice of estimators, shape restrictions, and the hurdle specification}
  \label{sec:method_choice}                                                                                                    
                                                                   
  Our empirical objective is to estimate the decomposition implied by the priority-pricing model: observed fees reflect both
  (i) a delay technology mapping priority and network state into expected confirmation delay and (ii) heterogeneous preferences
   over delay. This objective motivates a two-stage design that treats the delay technology as a nuisance function estimated
  flexibly, while keeping the fee equation structured enough to support interpretation and counterfactual analysis.

  \paragraph{First stage: flexible delay technology.}
  The delay technology is generated by a fee-prioritized queue with stochastic block arrivals, time-varying capacity, and
  heterogeneous transaction packages. Rather than impose a global parametric form, we estimate the conditional mean delay
  function using a Random Forest trained on an 80/20 random split of the data. The fitted model is then used to generate
  predicted delays for all observations, including those in the training set. Because the second stage includes epoch fixed
  effects that absorb aggregate time variation, and because the delay gradient---not the level of predicted delay---is the
  generated regressor of interest, the bias from in-sample prediction is attenuated in practice. We return to temporal
  stability and out-of-sample performance in Section \ref{sec:results:temporal}.

  \paragraph{Shape restriction: monotonicity in priority.}
  Queuing structure implies that, holding network state fixed, higher priority should not increase expected delay. Flexible
  learners need not respect this restriction. We therefore enforce monotonicity in the priority dimension when constructing the
   delay gradient that enters the second stage. For each epoch, we hold non-priority features at their epoch median, evaluate
  the fitted Random Forest over a fine priority grid, and apply isotonic regression to obtain a monotone decreasing schedule.
  Finite differences of this smooth schedule yield the local slope $\hat{W}'_{it}$ at each transaction's observed priority.
  This epoch-median approach isolates the causal effect of priority from variation in congestion features within the epoch, and
   avoids amplifying first-stage noise through per-observation partial derivatives.

  \paragraph{Second stage: the fee equation.}
  Transaction fees exhibit a large mass near the relay minimum and a steep right tail. In principle, a hurdle specification---separating the decision to pay a nontrivial fee from the amount paid---could capture both margins. In
  practice, fee payment is near-universal: fewer than $0.01\%$ of confirmed transactions fall below the effective minimum fee
  rate, reflecting the consensus rule that miners select by fee rate. We therefore estimate only the intensive margin, modeling $\log(\text{fee rate})$ conditional on positive fees via OLS with epoch fixed effects. The near-degenerate extensive margin contributes negligibly and is not separately modeled.

  \paragraph{State-based controls and counterfactual analysis.}
  Counterfactual exercises require a representation of the common fee environment that can be manipulated under alternative
  network conditions. Fully saturated epoch fixed effects absorb common shocks in-sample but do not map directly to
  counterfactual states. We therefore augment epoch effects with predetermined state variables capturing congestion, capacity,
  and block-arrival conditions---specifically, blockspace utilization, log time since the last block, and log mempool size.
  This preserves the role of high-frequency common shocks while explicitly estimating how observables shift the fee
  environment, which is the component we vary in counterfactual scenarios.

\subsection{First stage: estimating the delay technology}
\label{sec:firststage}

Priority is summarized by a transaction’s position in the fee-rate ordering among transactions that enter the relevant ranking set in epoch $t$. Let $N_t$ denote the number of transactions in this set. Because fee-rates are discretized (e.g., in sat/vB) and package behavior can generate identical fee-rates, ties are common. Since ties correspond to multiple distinct transactions competing for inclusion, we construct a tie-aware percentile rank that preserves the full transaction count.

\paragraph{Tie-aware fee-rate percentile.}
Let $r_{it}$ denote the fee-rate of transaction $i$ in epoch $t$. Define the fee-rate percentile

\begin{equation}p_{it}
\equiv
\frac{
\#\{j: r_{jt}<r_{it}\}
+\tfrac{1}{2}\#\{j: r_{jt}=r_{it}\}
}{N_t},
\label{eq:percentile_tieaware}
\end{equation}
so that $p_{it}\in(0,1)$ and all transactions with identical fee-rates receive the same percentile. This definition corresponds to the empirical cdf evaluated at the midpoint of the mass of ties and avoids boundary values that complicate finite differences.

\paragraph{Delay technology.}
We model confirmation delay as a flexible function of priority and predetermined system state:
\begin{equation}
W_{it} = g(p_{it}, Z_{it})+u_{it}.
\label{eq:firststage_g}
\end{equation}

The function $g(\cdot)$ is estimated using a random forest regressor, which accommodates sharp nonlinearities and interactions characteristic of fee-prioritized queues with stochastic block arrivals and time-varying capacity. 

\paragraph{Cross-fitting.}
Because the second stage uses generated regressors, we construct first-stage predictions out of sample using epoch-level cross-fitting. We partition the sample into $K$ folds at the epoch level, estimate the forest on $K-1$ folds, and predict delay in the held-out fold. Let $\widehat g^{(-k)}(\cdot)$ denote the fitted model excluding fold $k$. For observations $(i,t)$ in fold $k$, define the out-of-sample predicted delay
\begin{equation}
\widehat W_{it}\equiv\widehat g^{(-k)}(p_{it}, Z_{it}).
\label{eq:What_raw}
\end{equation}

\paragraph{Monotonicity in priority.}
Queueing structure implies that, holding state fixed, higher priority should not increase expected delay. Since standard random forests do not impose shape restrictions, we enforce monotonicity by post-processing out-of-sample forest predictions along the priority dimension.

Fix an observation $(i,t)$ and let $z_{it}$ collect the non-priority covariates in \eqref{eq:firststage_g}. Evaluate the fitted forest on a fine grid of percentiles $\mathcal{P}=\{p^{(1)},\ldots,p^{(M)}\}\subset(0,1)$:
\[
\widehat m^{(-k)}_{it}(p^{(m)})\equiv \widehat g^{(-k)}(p^{(m)},Z_{it}),\qquad m=1,\ldots,M.
\]
We compute the closest weakly decreasing sequence to $\{\widehat m^{(-k)}_{it}(p^{(m)})\}_{m=1}^M$ in squared error loss via isotonic regression. Let $\widetilde m^{(-k)}_{it}(p^{(m)})$ denote the fitted monotone sequence, and define the monotone first-stage predictor $\widetilde g^{(-k)}(p,z_{it})$ as the linear interpolation of the grid points

\[\{(p^{(m)},\widetilde m^{(-k)}_{it}(p^{(m)}))\}_{m=1}^M.\]

We replace \eqref{eq:What_raw} by the monotone prediction

\[\widetilde W_{it}\equiv \widetilde g^{(-k)}(p_{it},Z_{it}).\]

\paragraph{Local slope of the delay schedule.}
The fee condition in Section 3 depends on how expected delay changes locally with priority. We approximate the local slope using a symmetric finite difference computed from the monotone prediction function:

\begin{equation}
\widetilde W'_{it}\equiv\frac{\widetilde g^{(-k)}(\bar {p}_{it}^{+},Z_{it})\widetilde g^{(-k)}(\bar p_{it}^{-},Z_{it})}{\bar p_{it}^{+}\bar p_{it}^{-}},\qquad\bar p_{it}^{\pm}\equiv \min\{1-\underline p,\max\{\underline,p_{it}\pm \Delta\}\}.
\label{eq:Wprime_fd}
\end{equation}

Here $\underline p>0$ is a small trimming constant to avoid boundary artifacts and $\Delta>0$ is a bandwidth parameter. In the baseline we set $\underline p=0.01$ and $\Delta=0.01$ and verify robustness to alternative choices.

\paragraph{Robustness check: alternative monotone first-stage learner.}
To verify that the second-stage results are not driven by the particular learner used in the first stage, we repeat the first stage using a monotone gradient-boosted tree model that imposes a weakly decreasing constraint in $p_{it}$ by construction while allowing flexible interactions with $\widehat\rho_t$ and $\mathrm{Blockspace}_t$. Using the resulting out-of-sample slope estimates $\widehat W'^{\,\mathrm{MGB}}_{it}$, we re-estimate the second stage and show that the estimated impatience mapping and counterfactual implications are stable across first-stage specifications.

\subsection{Second stage: hurdle fee model with a monotone spline and state-based controls}
\label{sec:secondstage}

Observed fees exhibit a large mass near zero and a steep right tail. To accommodate this distributional feature while preserving the model-implied decomposition, we estimate a hurdle specification that separates (i) the extensive margin decision to pay a nontrivial fee from (ii) the fee level conditional on paying a nontrivial fee.

\paragraph{State-based controls and the common fee environment.}
Let $S_t$ collect predetermined, high-frequency measures of network state, such as congestion and capacity. We represent the common fee environment as

\[\eta_t = S_t'\theta + \xi_t,\]

where $\xi_t$ captures residual epoch variation orthogonal to $S_t$ and is normalized by $\mathbb{E}[\xi_t]=0$. This representation preserves the role of epoch-level common shocks while identifying an explicit component, $S_t'\theta$, that can be manipulated in counterfactual exercises by varying $S_t$.








We estimate a log-fee equation:

\begin{equation}
\log b_{it}=\alpha_0+s(\iota_{it})+\alpha_1 \log \widetilde W'_{it}+X_{it}'\beta+S_t'\theta+\xi_t+\varepsilon_{it},\qquad \text{for}\ \ b_{it}>\epsilon.
\label{eq:intensive_log}
\end{equation}

The coefficient $\alpha_1$ captures the sensitivity of fees to the local steepness of the delay schedule, as implied by the priority-pricing condition in Section 3.

\paragraph{Monotone spline for the impatience mapping.}
Economic structure implies that greater impatience should not reduce either (i) the propensity to pay a fee or (ii) the fee level among fee-paying transactions. We therefore estimate $h(\cdot)$ and $s(\cdot)$ under monotonicity restrictions using an I-spline basis. Let $\{\kappa_1,\ldots,\kappa_J\}$ denote knots placed at predetermined impatience points (baseline: quantiles of $\iota_{it}$ in the fee-relevant region with additional mass in the upper tail). Let $\{I_1(\iota),\ldots,I_L(\iota)\}$ denote the associated I-spline basis functions, which are weakly increasing in $\iota$. We parameterize

\begin{align}
s(\iota) &= \delta^{(s)}_0 + \sum_{\ell=1}^{L}\delta^{(s)}_\ell I_\ell(\iota), \qquad \delta^{(s)}_\ell \ge 0 \ \ \forall \ell\ge 1.
\label{eq:monotone_s}
\end{align}

The nonnegativity constraints imply that both $h(\cdot)$ and $s(\cdot)$ are weakly increasing.

\paragraph{Estimation.}

We estimate equation \eqref{eq:intensive_log} by constrained least squares with the monotonicity constraints in \eqref{eq:monotone_s}. Standard errors are computed allowing for sampling variability in the generated regressor $\log \widetilde W'_{it}$ by using  cross-fitting folds and clustering at the epoch level.

\paragraph{Counterfactual aggregation.}
For a counterfactual state $S_t^{cf}$, the hurdle model delivers two components: the predicted selection probability

\[\widehat \pi^{cf}_{it}\equiv \Pr(H_{it}=1\mid \iota_{it},\widetilde W'_{it},X_{it},S_t^{cf},\xi_t=0),\]

and the predicted conditional mean fee among fee-paying transactions. If \eqref{eq:intensive_log} is used to predict $m^{cf}_{it}\equiv \mathbb{E}[b_{it}\mid b_{it}>\epsilon,\cdot]$, we compute

\[\widehat m^{cf}_{it}=\exp\!\left(\widehat \mu^{cf}_{it}\right)\cdot \widehat \psi,\qquad\widehat\mu^{cf}_{it}=\widehat \alpha_0+\widehat s(\iota_{it})+\widehat \alpha_1 \log \widetilde W'_{it}+X_{it}'\widehat \beta+(S_t^{cf})'\widehat\theta,\]

where $\widehat \psi$ is a smearing factor estimated from the second-stage residuals. The implied unconditional expected fee is then

\[\mathbb{E}[b_{it}\mid \cdot]\approx\widehat \pi^{cf}_{it}\cdot \widehat m^{cf}_{it}+(1-\widehat \pi^{cf}_{it})\cdot \mathbb{E}[b_{it}\mid b_{it}\le\epsilon],
\]

where the final term is economically negligible under the baseline choice of $\epsilon$ and is set to the sample mean below $\epsilon$ unless otherwise stated.

\subsection{Identification and inference}
\label{subsec:identification}
\textbf{Identification} follows the maintained assumption that individual transactions are atomistic relative to the mempool: each transactor takes the prevailing mapping from priority to delay as given at the time of fee choice. Under this assumption, $\widehat W_{it}$ summarize the relevant technological constraint faced by transaction $i$ in epoch $t$. The spline component $s(\iota_{it})$ provides a flexible approximation to the mapping from impatience into fees. \textbf{Inference} is implemented by computing standard errors clustered at the epoch level to allow for within-epoch dependence. To incorporate first-stage estimation uncertainty, we also report an epoch-block bootstrap that resamples epochs and re-estimates both stages within each resample. See Appendix \ref{app:identification} for a formal analysis of identification and inference.

\section{Results}                                                                                                            
  \label{sec:results}                                                                                                          
                                                                                                                               
  \subsection{Data}
  \label{sec:results:data}

  Our sample covers \num{11110277} confirmed Bitcoin transactions observed
  between 3~August 2025 and 9~October 2025 (approximately ten weeks), drawn
  from a continuous mempool observer database.
  Transactions are grouped into \num{1988} non-overlapping 30-minute
  \emph{epochs}, which serve as the unit of observation for epoch fixed effects
  and clustered inference throughout the analysis.

  \paragraph{Weight correction.}
  Raw transaction weights recorded by the local Bitcoin node are systematically
  understated for SegWit inputs: the node counts witness bytes at full weight
  rather than the discounted rate.
  We correct this by merging with independently fetched data from the
  mempool.space API, which reports the true virtual size (\texttt{vsize}) and
  witness-adjusted weight.
  This correction is available for $99.2\%$ of transactions and raises the mean
  weight by $+272.5$~weight units~(WU), a $28.8\%$ average increase.
  Fee rates derived from corrected weights are materially lower for SegWit-heavy
  transaction types and are the figures used throughout.

  \paragraph{CPFP package collapsing.}
  Child-pays-for-parent (CPFP) arrangements create a mismatch between the fee
  rate a transaction \emph{appears} to offer and the rate that actually
  determines its priority in the miner's mempool.
  Following standard practice, we identify parent--child pairs and collapse each into a single \emph{package} observation whose fee rate equals the total
  package fee divided by the total package virtual size.
  This affects approximately $24.4\%$ of the final sample.

  \paragraph{Transaction composition.}
  We annotate each transaction with a structural label using a rule-based
  classifier.  Table~\ref{tab:tx_labels} reports the resulting distribution.
  The majority are simple payment transactions; data-carrying transactions
  (carrying \texttt{OP\_RETURN} outputs or inscriptions) account for roughly
  one fifth of the sample.

  \begin{table}[h]
    \centering
    \caption{Transaction label distribution (representative sample)}
    \label{tab:tx_labels}
    \begin{tabular}{lr}
      \toprule
      Label & Share (\%) \\
      \midrule
      Normal payment   & 69.6 \\
      Data-carrying    & 19.2 \\
      Consolidation    & \phantom{0}7.4 \\
      Batch payment    & \phantom{0}3.7 \\
      CoinJoin         & \phantom{0}0.2 \\
      \midrule
      Total (annotated) & 100.0 \\
      \bottomrule
    \end{tabular}
    \par\smallskip
    {\footnotesize \textit{Note}: Shares computed from a labelled subsample;
     ``unknown'' labels excluded.}
  \end{table}

  Additional sample characteristics: RBF-flagged transactions account for
  $5.3\%$ of observations; $24.4\%$ are CPFP packages; the fee rate
  distribution is heavily right-skewed, with a median of approximately
  $5$~sat/vB and a long tail of high-urgency transactions exceeding
  $100$~sat/vB.

  \subsection{Stage 1: Estimating the Delay Technology}
  \label{sec:results:stage1}

  The first stage estimates the \emph{delay technology}
  $\hat{W}(p_{it}, s_t)$: the expected log confirmation delay as a function of
  a transaction's priority percentile $p_{it}$ within its epoch and the
  prevailing mempool state~$s_t$.
  We use a Random Forest regressor because it flexibly captures the non-linear,
  interaction-heavy relationship between priority and delay without imposing a
  parametric form.

  \paragraph{Model specification.}
  The response variable is $\log(\text{waittime} + 1)$ in seconds.
  The feature set is $\{p_{it},\; \text{blockspace utilization},\;
  \text{mempool size},\; \text{mempool tx count}\}$, where priority percentile
  $p_{it}$ is computed within each epoch using a tie-aware rank so that the
  distribution is approximately uniform on $[0,1]$ within every 30-minute
  window.\footnote{The mempool features enter the Random Forest in levels rather
  than logs.  Because tree-based methods are invariant to monotone
  transformations of individual features, the choice is immaterial for the
  fitted predictions.}
  We use 200~trees, maximum depth~15, and minimum leaf size~20.
  A random 80/20 train--test split is used for evaluation.

  \paragraph{Performance.}
  On the held-out test set the model achieves $R^2 = 0.611$,
  RMSE $= 0.812$ (in log-seconds).
  The remaining unexplained variance is expected: confirmation timing depends on
  factors not observable in the mempool (inter-block intervals, miner selection
  policies, network propagation), so the model deliberately captures the
  \emph{systematic} component attributable to priority and congestion.

  \begin{figure}[H]
    \centering
    \includegraphics[width=\linewidth]{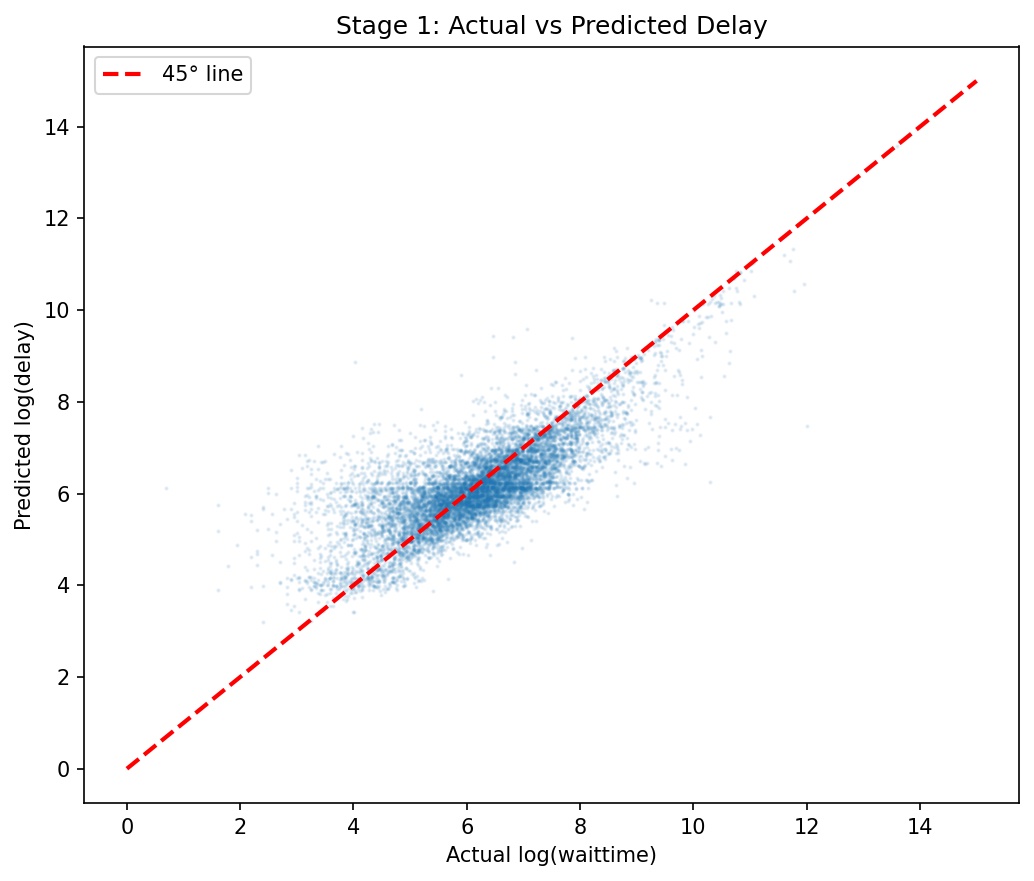}
    \caption{Actual vs.\ predicted $\log(\text{waittime})$ for a random
             $10{,}000$-observation subsample; the dashed line is the 45-degree
             reference.}
    \label{fig:stage1_diag}
  \end{figure}

  \paragraph{Feature importance.}
  Blockspace utilization dominates, accounting for $52.0\%$ of total feature
  importance in the fitted forest, followed by mempool transaction count
  ($24.6\%$), mempool size ($13.7\%$), and priority percentile ($9.7\%$).
  This ordering confirms that the level of network congestion---both within the
  current block and in the broader mempool---is the primary determinant of
  absolute confirmation delays, with a transaction's relative priority
  controlling the residual variation within a given congestion regime.
  Figure~\ref{fig:stage1_importance} illustrates the importances.

  \begin{figure}[h]
    \centering
    \includegraphics[width=\linewidth]{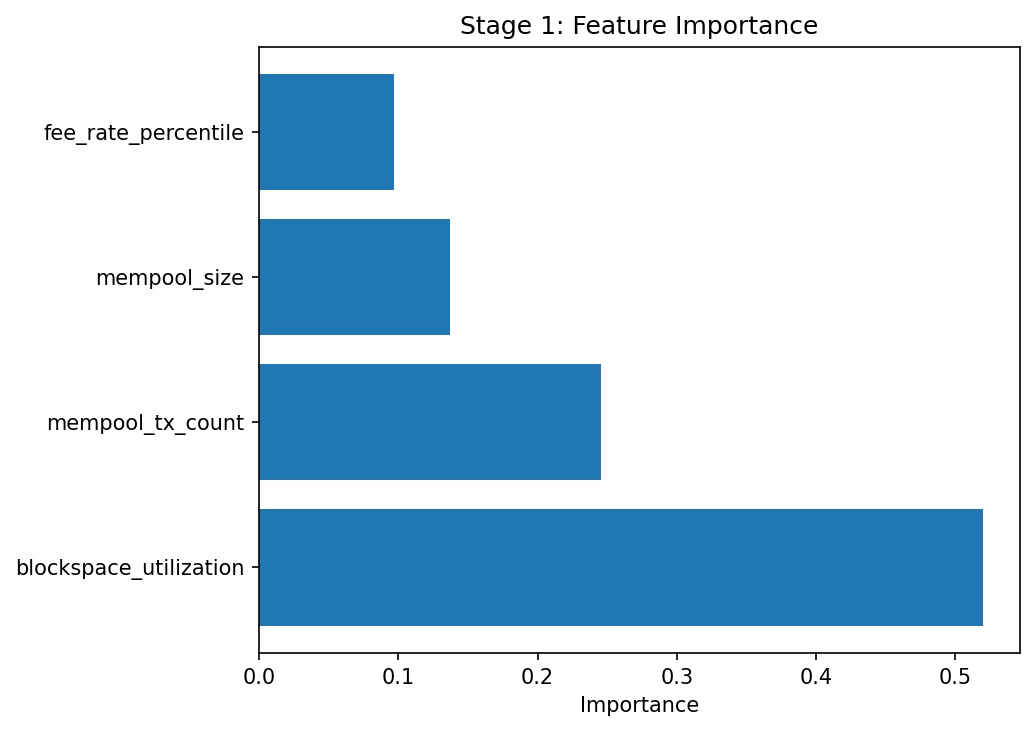}
    \caption{Random Forest feature importances for the Stage~1 delay technology
             model.  Blockspace utilization and mempool depth jointly capture
             aggregate congestion, while priority percentile captures a
             transaction's competitive position within the queue.}
    \label{fig:stage1_importance}
  \end{figure}

  \paragraph{Monotonicity enforcement and local slopes.}
  The VCG identification argument requires that $\hat{W}$ is weakly decreasing
  in priority~$p$: higher-priority transactions must (weakly) confirm faster.
  For each epoch, we hold the non-priority features at their epoch median,
  evaluate the fitted forest over a fine 99-point priority grid, and apply
  isotonic regression to obtain a monotone decreasing schedule.
  The key input to Stage~2 is the local slope
  $\hat{W}'_{it} = \partial \hat{W}/\partial p$, which we compute via symmetric
  finite differences with step size $\delta = 0.05$ on a linear interpolant of
  this monotone schedule.
  All slopes are constrained to be positive (higher priority $\Rightarrow$
  weakly lower delay).
  Figure~\ref{fig:stage1_congestion} shows the resulting delay--priority
  relationship stratified by congestion regime: the gradient is steeper during
  high-congestion periods, implying that priority is more valuable when the
  mempool is full.

  \begin{figure}[h]
    \centering
    \includegraphics[width=\linewidth]{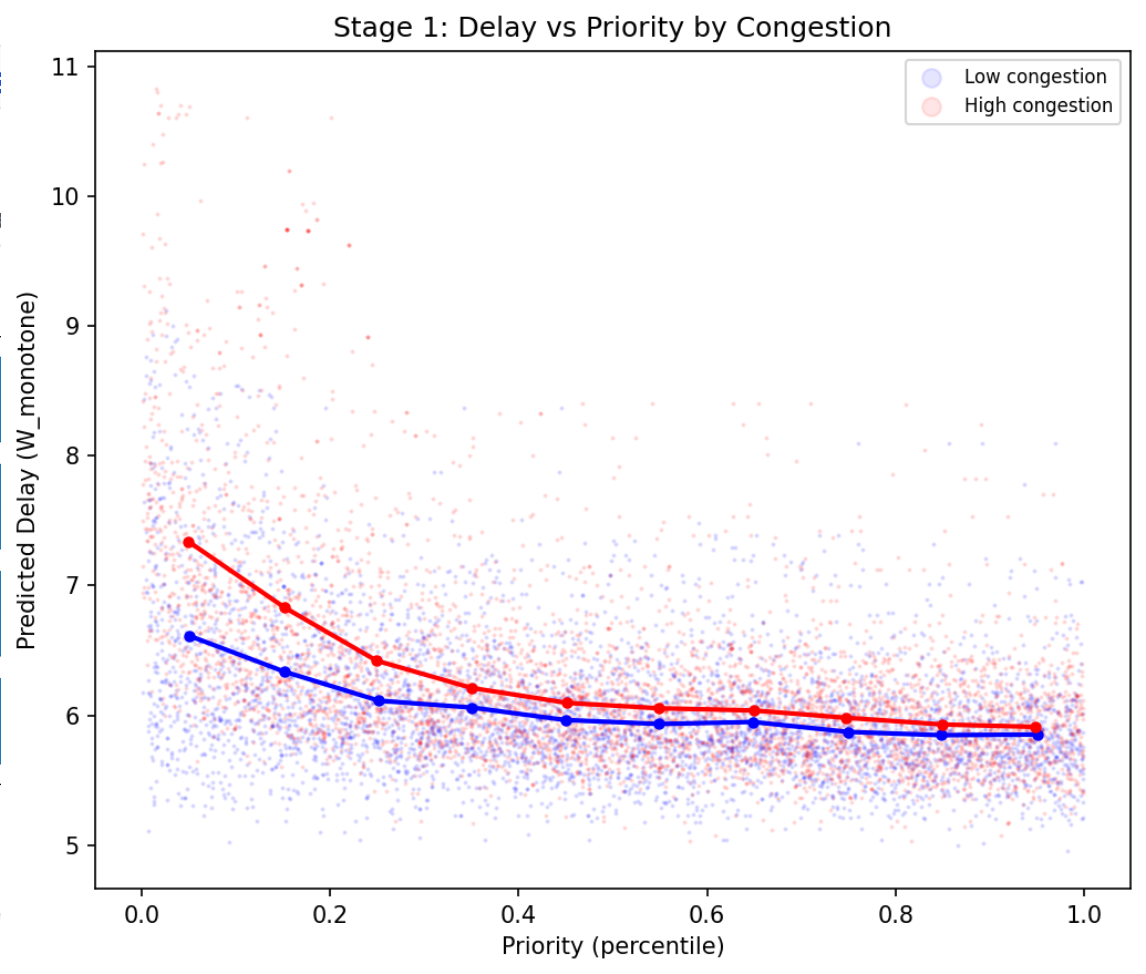}
    \caption{Delay--priority relationship stratified by mempool congestion
             regime.  The delay gradient $\hat{W}'$ is steeper during
             high-congestion periods, implying that priority is more valuable
             when the mempool is full.}
    \label{fig:stage1_congestion}
  \end{figure}
\paragraph{Delay gradient regimes.}
  The epoch-median sweep reveals that the delay technology is not uniformly
  informative.  Figure~\ref{fig:trivial_nontrivial} shows that $79.8\%$ of
  epochs ($7{,}980$ of $10{,}000$) exhibit a \emph{trivial} delay gradient:
  the monotone-enforced schedule is essentially flat across the priority grid,
  so $\hat{W}'_{it} \approx 0$ for all transactions in those epochs.  Only the
  remaining $20.2\%$ of epochs display an \emph{active} delay technology in
  which higher priority produces a measurably steeper reduction in expected
  delay.  The right panel confirms that the fee-rate distributions of the two
  regimes largely overlap, with the non-trivial regime showing a slightly
  heavier right tail.

  This pattern has a natural interpretation: during low-congestion periods the
  mempool clears quickly regardless of priority, so the Random Forest correctly
  learns a flat delay schedule---there is no priority gradient to price.  The
  structural channel identified by the VCG model is therefore operative
  primarily during congestion episodes, which is precisely when fee
  differentiation matters most.  In Stage~2, observations from trivial-gradient
  epochs still contribute to the estimation of control-variable and epoch
  fixed-effect coefficients, but the identifying variation for $\alpha_1$ is
  concentrated in the $20\%$ of epochs where the delay technology is active. We retain all epochs in the estimation rather than restricting to the non-trivial regime, so that the model is disciplined by the full data distribution and the coefficient estimates are not conditional on an ex-post sample selection that could inflate the apparent strength of the structural channel.

  \begin{figure}[h]
    \centering
    \includegraphics[width=\linewidth]{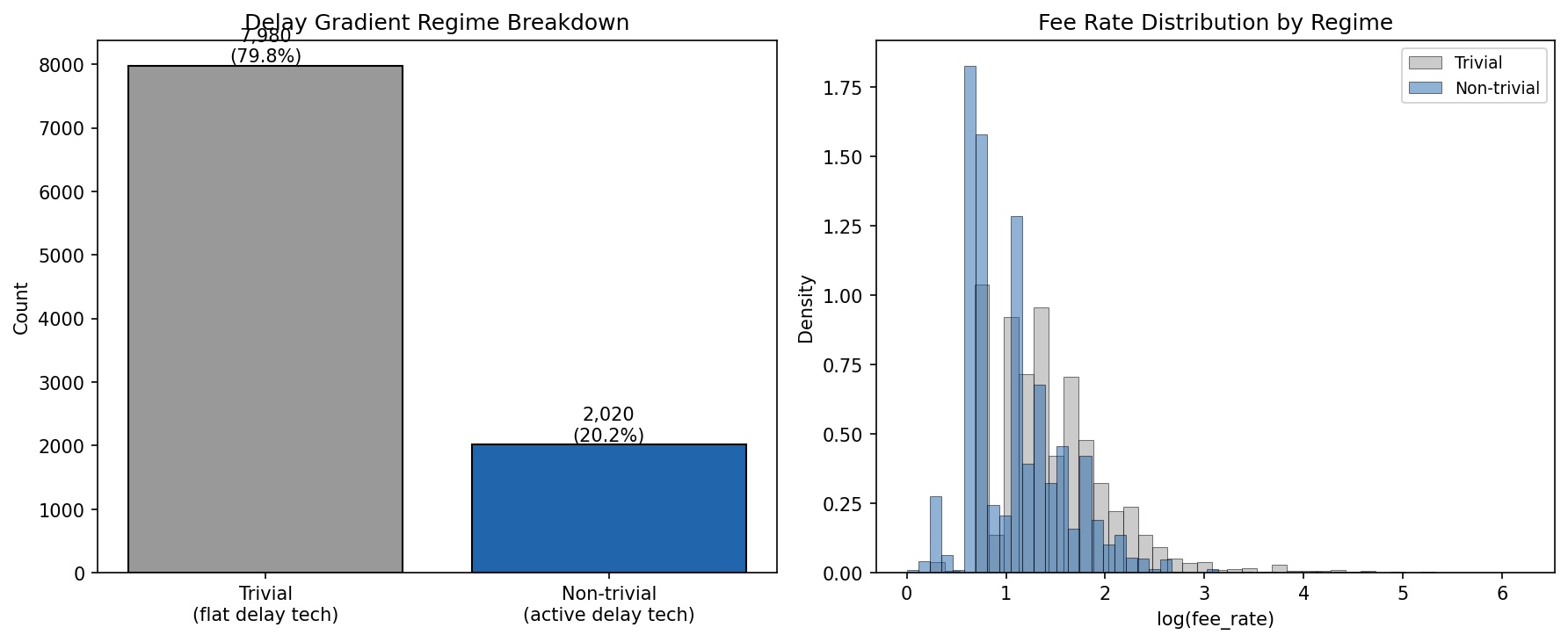}
    \caption{\emph{Left}: Breakdown of epochs by delay gradient regime.
             Trivial epochs (flat delay technology, $\hat{W}' \approx 0$)
             account for $79.8\%$ of the sample; the structural priority
             channel is active in only $20.2\%$ of epochs.
             \emph{Right}: Log fee-rate distributions by regime.  The two
             distributions largely overlap, with the non-trivial regime
             exhibiting a slightly heavier right tail consistent with
             congestion-driven fee differentiation.}
    \label{fig:trivial_nontrivial}
  \end{figure}
 
  \subsection{Stage 2: Fee Equation Estimation}
  \label{sec:results:stage2}

  We estimate

  \begin{equation}
    \log b_{it} \;=\;
    \alpha_0 + \alpha_1 \log \hat{W}'_{it}
    + X_{it}'\beta + S_t'\theta + \xi_t + \varepsilon_{it},
    \label{eq:intensive}
  \end{equation}

  \noindent where $\hat{W}'_{it}$ is the Stage~1 delay gradient,
  $X_{it}$ is a vector of transaction-level controls (RBF flag, CPFP status,
  log output amount, input/output counts, and type indicators), $S_t$ collects
  block-state variables (blockspace utilization, log time since last block, log
  mempool size), and $\xi_t$ is a full set of epoch fixed effects.
  Standard errors are clustered at the epoch level using the Liang--Zeger
  sandwich estimator.
  The impatience term $s(\iota_{it})$ that appears in the theoretical
  specification~\eqref{eq:intensive_log} is intentionally excluded here: the
  impatience proxy (derived from re-spend intervals) has incomplete coverage and
  its inclusion does not materially affect the coefficient on
  $\log \hat{W}'_{it}$; we treat the impatience channel as a robustness check
  rather than a baseline result.

  \paragraph{Results.}
  Table~\ref{tab:intensive} reports coefficient estimates for the base features.
  All variables are statistically significant at the $0.1\%$ level.

  \begin{table}[h]
    \centering
    \caption{Stage 2 estimates --- equation~\eqref{eq:intensive}}
    \label{tab:intensive}
    \begin{threeparttable}
      \begin{tabular}{l
                      S[table-format=-1.4]
                      S[table-format=1.4]
                      S[table-format=-3.2]
                      @{\hspace{0.8em}}l}
        \toprule
        Variable & {Coef.} & {Cl.\ SE} & {$t$-stat} & {Sig.} \\
        \midrule
        \multicolumn{5}{l}{\emph{Structural variable}} \\
        $\log \hat{W}'$ (delay gradient= $-D$ )     & -0.0459 & 0.0008 &  -58.11 & *** \\[4pt]
        \multicolumn{5}{l}{\emph{Transaction controls}} \\
        RBF flag                             &  0.6568 & 0.0147 &   44.80 & *** \\
        CPFP package                         & -0.1622 & 0.0046 &  -35.00 & *** \\
        $\log(\text{total output, sat})$     &  0.0422 & 0.0007 &   63.31 & *** \\
        $\log(\text{inputs})$                & -0.0984 & 0.0018 &  -55.22 & *** \\
        $\log(\text{outputs})$               & -0.0734 & 0.0016 &  -45.64 & *** \\
        Has \texttt{OP\_RETURN}              & -0.2287 & 0.0177 &  -12.91 & *** \\
        Has inscription                      & -0.1315 & 0.0204 &   -6.45 & *** \\[4pt]
        \multicolumn{5}{l}{\emph{Block-state controls}} \\
        Blockspace utilization               &  0.2497 & 0.0073 &   34.25 & *** \\
        $\log(\text{time since last block})$ &  0.1499 & 0.0037 &   40.05 & *** \\
        $\log(\text{mempool size})$          &  0.0699 & 0.0094 &    7.46 & *** \\
        \midrule
        Intercept                            & -1.9932 & 0.1879 &  -10.61 & *** \\
        \midrule
        \multicolumn{5}{l}{Epoch fixed effects: \num{1987} dummies (not shown)} \\
        \bottomrule
      \end{tabular}
      \begin{tablenotes}[flushleft]
        \footnotesize
        \item \textit{Notes}: $N$ = 11,110,277 observations; 1,988 epoch clusters (30-min windows).
          Standard errors clustered at the epoch level (Liang--Zeger sandwich); mean SE inflation relative to na\"ive OLS: $6.90\times$. Degrees of freedom for $t$-statistics: $1987$.
          In-sample $R^2 = 0.435$.
          Smearing factor (Duan, 1983) for retransformation: $1.344$.
          *** $p < 0.001$.
      \end{tablenotes}
    \end{threeparttable}
  \end{table}

  \paragraph{Structural interpretation.}
  The negative and precisely estimated coefficient on $\log \hat{W}'_{it}$
  ($\hat{\alpha}_1 = -0.046$, $t = -58.1$) is the central result of the paper.\footnote{A negative coefficient value on $\log \hat{W}'_{it}$ is equivalent to a positive coefficient on the elasticity of the slope of the value of reduction in delay, $D^{'}(p_{it}, s_{t})$. }
  The VCG theory predicts that a transaction's fee rate should be
  \emph{decreasing} in the marginal delay reduction it generates: when each
  additional unit of priority produces a large reduction in expected wait time
  (steep delay gradient), senders can achieve a given confirmation speed at a
  \emph{lower} fee.  Conversely, when the delay schedule is flat---priority does
  little to accelerate confirmation---senders must bid higher to separate
  themselves from the queue.
  The estimate is consistent with this prediction and survives the inclusion of
  almost two thousand epoch fixed effects, indicating it is not driven by
  aggregate time trends.

  \paragraph{Transaction-type controls.}
  RBF-flagged transactions pay a premium of approximately $93\%$
  ($\exp(0.657) \approx 1.93$), suggesting that opting into replace-by-fee
  signals high urgency and is compensated accordingly by miners.
  CPFP packages pay $15\%$ less than stand-alone transactions of the same
  apparent priority, consistent with the bundled nature of the fee subsidy.
  Transactions carrying \texttt{OP\_RETURN} outputs or inscriptions pay $20\%$
  and $12\%$ lower fee rates respectively, reflecting their lower time
  sensitivity.

  \paragraph{Block-state controls.}
  Blockspace utilization (fraction of the current block's weight already filled
  at the time of observation) and log time elapsed since the previous block both
  have positive, economically meaningful coefficients.
  A block that is $10$~percentage points closer to full is associated with a
  $2.5\%$ higher fee rate; each doubling of the inter-block interval adds
  roughly $10\%$ to fee rates, consistent with impatience growing as the gap
  since the last confirmation widens.
  Mempool depth (log mempool size in bytes) also enters positively, capturing
  longer-run congestion not fully absorbed by the priority percentile.

  \paragraph{Epoch fixed effects and inference.}
  Epoch fixed effects absorb all variation in unobserved aggregate conditions
  within a 30-minute window---miner preferences, network propagation, and macro
  Bitcoin price dynamics---ensuring that the structural coefficients identify
  within-epoch, cross-transaction variation.
  The mean SE inflation of $6.90\times$ relative to na\"ive OLS confirms that
  transactions within the same epoch are substantially correlated, making epoch
  clustering essential for valid inference.

  \subsubsection{Model diagnostics}
  \label{sec:results:stage2:diag}

  Figures~\ref{fig:stage2_act} and~\ref{fig:stage2_diag} present standard diagnostics for the second-stage regression. The fitted-versus-actual plot shows good fit over the bulk of the distribution, with underprediction primarily in the extreme upper tail. The visible vertical streaking further suggests that realized fee rates may cluster at a discrete set of focal values rather than vary continuously. A plausible hypothesis is that wallet interfaces often implement a small menu of common fee options—e.g., 1, 1.5, 2.5 sat/vB—inducing bunching in realized fees even conditional on differing transaction characteristics and mempool states. If so, the diagonal bands in the residuals-versus-fitted plot (right panel) are simply the corresponding geometric implication, because residuals equal realized minus fitted fees. This same mechanism may contribute to the modest heteroskedasticity at low fitted values, where realized fees are additionally constrained by an effective lower bound. The residual histogram is nonetheless approximately normal, with slight positive skew, consistent with the smearing correction used in retransformation.
  
\begin{figure}[H]
    \centering
    \includegraphics[width=\linewidth/2]{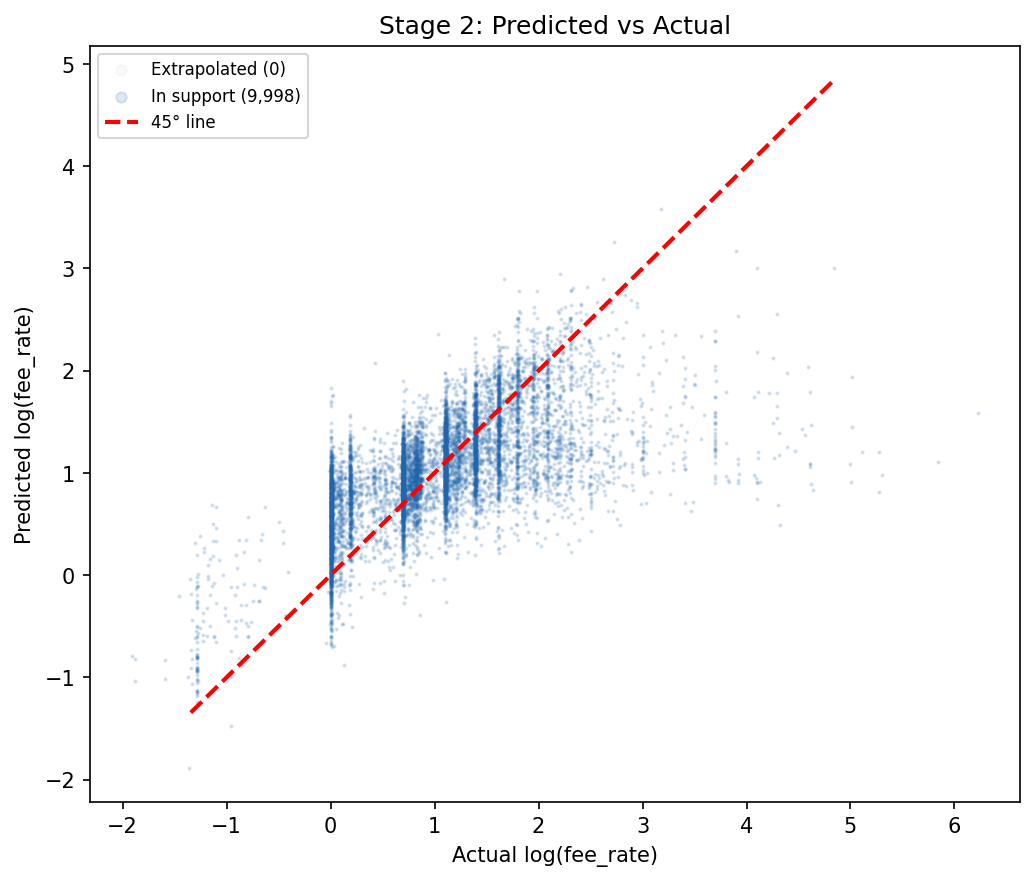}
    \caption{Actual vs. predicted log fee rate with a 10k subsample plotted.}
    \label{fig:stage2_act}
  \end{figure}
  
  \begin{figure}[H]
    \centering
    \includegraphics[width=\linewidth]{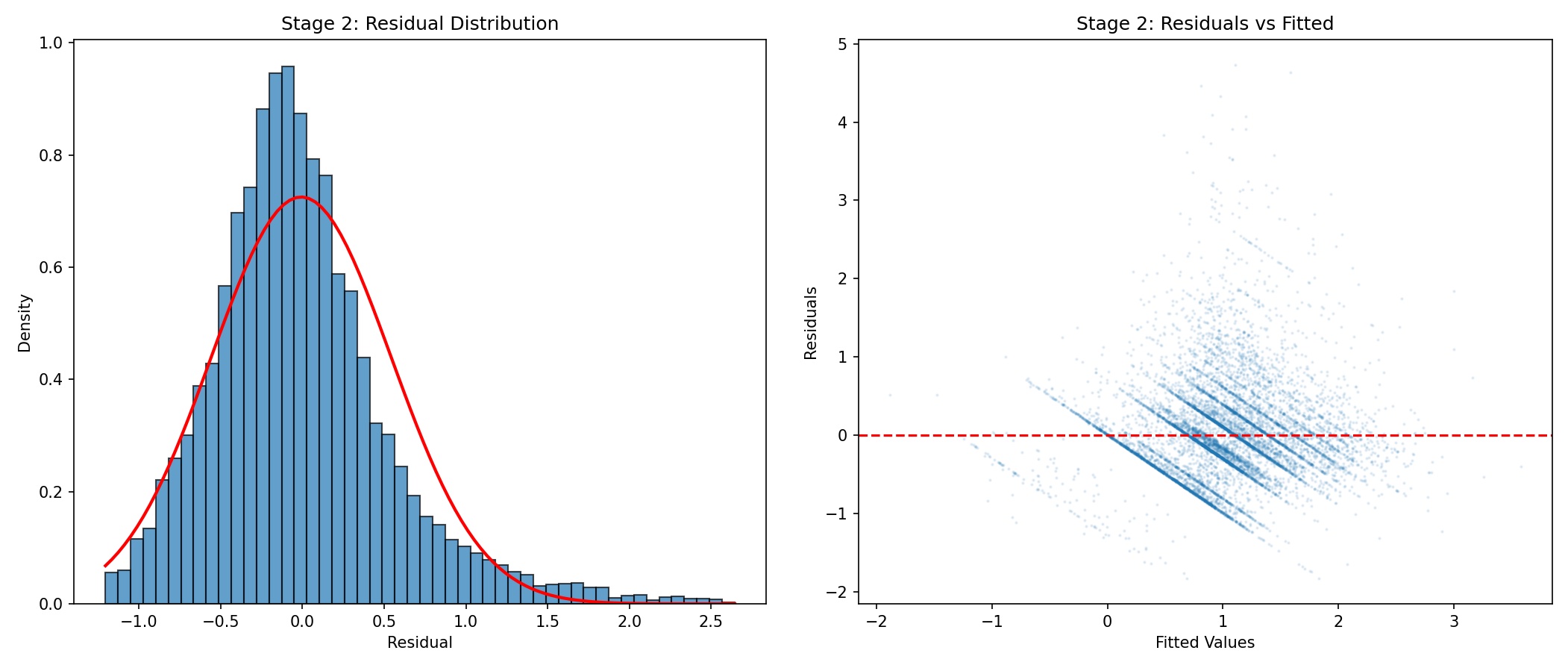}
    \caption{\emph{Left}: Residual density with a fitted normal overlay.
             \emph{Right}: Residuals vs.\ fitted values.}
    \label{fig:stage2_diag}
  \end{figure}

  Figure~\ref{fig:stage2_coef} displays the base-feature coefficient estimates
  with $95\%$ epoch-clustered confidence intervals.

  \begin{figure}[H]
    \centering
    \includegraphics[width=0.75\linewidth]{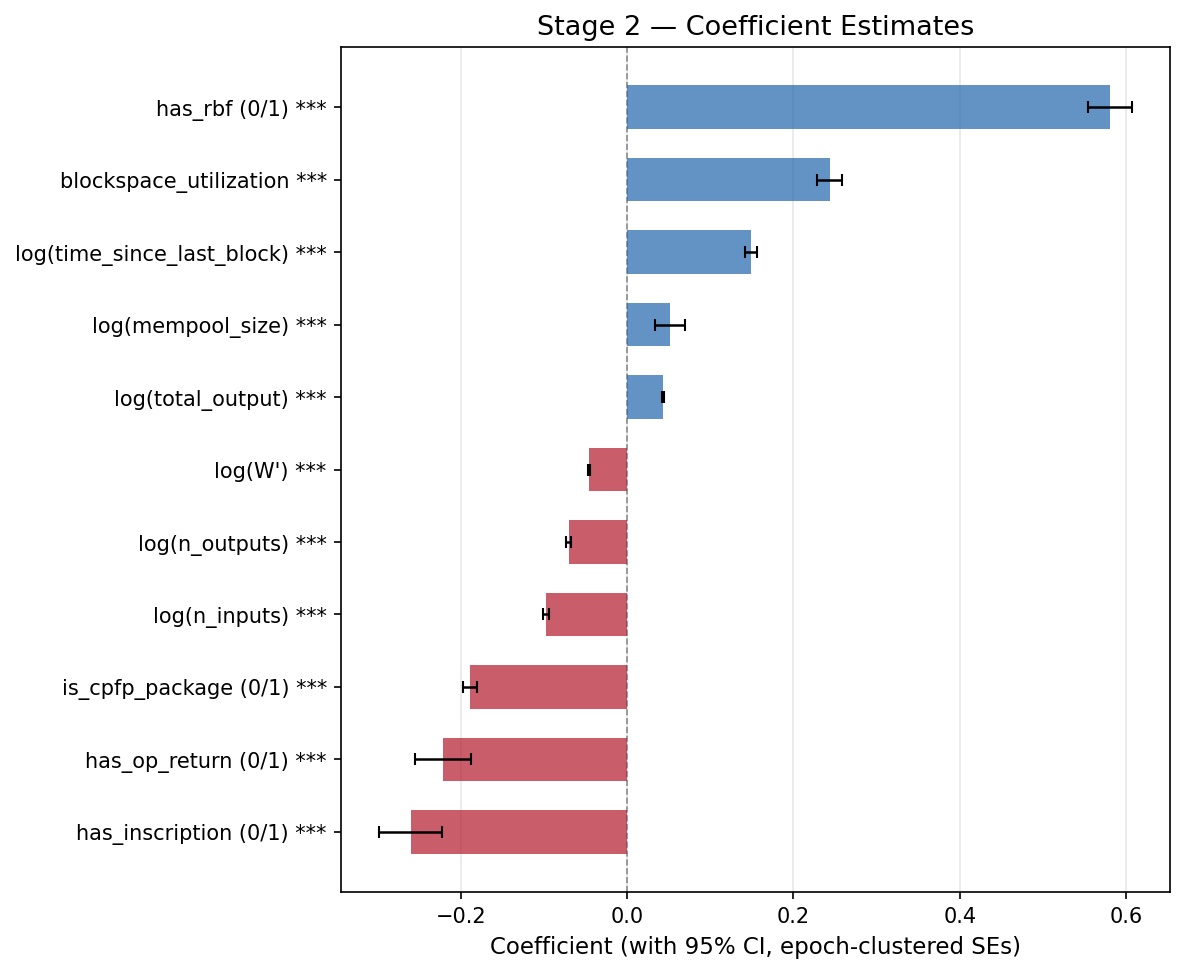}
    \caption{Second-stage coefficient estimates with $95\%$ confidence
             intervals (epoch-clustered standard errors).  Blue bars indicate
             positive effects on log fee rate; red bars indicate negative
             effects.  All base features are significant at $p < 0.001$.}
    \label{fig:stage2_coef}
  \end{figure}

\subsection{Robustness}
\label{subsec:robustness}

Appendix \ref{app:spline} reports two complementary robustness checks. First, allowing the impatience proxy to enter flexibly through a monotone I-spline leaves the estimated coefficient on the delay-gradient term essentially unchanged (-0.0455 versus -0.0459 in the baseline), improves fit only modestly, and implies an economically small aggregate impatience effect: the implied change in log fee rates from the median to the 95th percentile of impatience is 0.0270 (SE = 0.0032). Second, replacing the Random Forest first stage with a monotone gradient-boosted tree learner on the full sample yields $\hat{\alpha}_1$ = -0.0861 (SE = 0.0005). This is more negative in magnitude than the Random-Forest estimate, but it preserves the sign, precision, and qualitative interpretation of the structural channel. The corresponding aggregate impatience effect remains positive but smaller at 0.0178 (SE = 0.0023). Taken together, these results indicate that the main finding is not an artefact of the baseline functional-form choice: cross-sectional fee variation remains primarily tied to the steepness of the delay schedule, while impatience plays a secondary role, although the exact magnitude of the delay-gradient coefficient is somewhat sensitive to the first-stage learner.

  \subsection{Temporal Stability}
  \label{sec:results:temporal}
 
  A key question for any structural model estimated on time-series data is
  whether the coefficients represent stable structural parameters or are
  artifacts of a particular sample period.  We conduct six complementary
  diagnostics, each requiring re-estimation on different subsets of the data.

  \paragraph{Intraclass correlation.}
  The intraclass correlation coefficient~(ICC) for $\log \hat{W}'$ is $0.186$,
  indicating that roughly one-fifth of its total variance lies between epochs
  rather than within them.  The implied design effect is~$980$, reducing the
  effective sample size from \num{11108662} to approximately \num{11340}.
  This confirms that epoch-clustered standard errors are essential; na\"ive OLS standard errors would dramatically
  overstate precision.
 \begin{figure}[H]
    \centering
    \includegraphics[width=\linewidth]{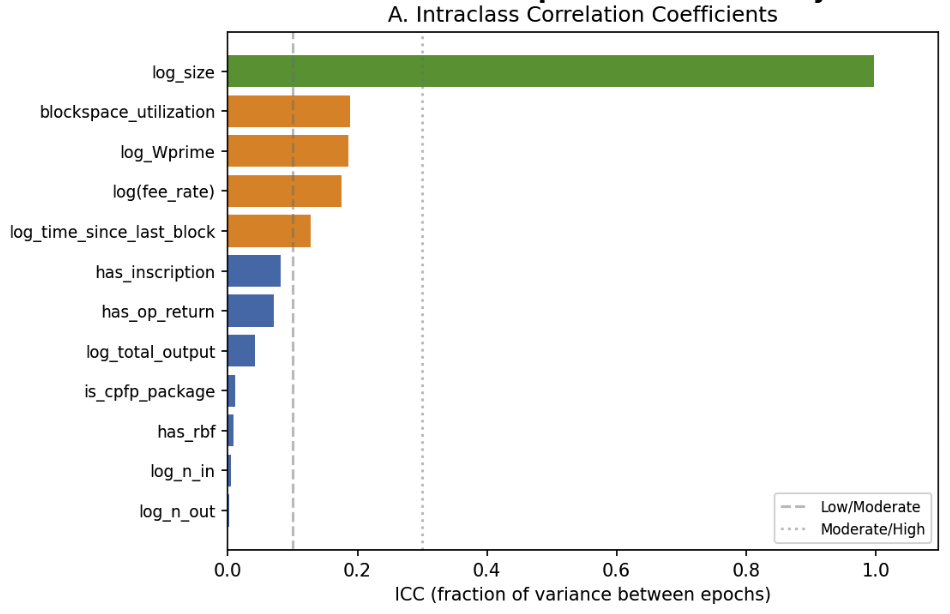}
    \caption{Intraclass correlation coefficients for all model
             features.}
      
    \label{fig:ICC}
  \end{figure}
  
  \paragraph{Cumulative precision.}
  Figure~\ref{fig:CPC} traces the clustered standard error of
  $\hat{\alpha}_1$ as epochs are added sequentially.  Precision improves
  monotonically, with the SE falling from $0.0016$ at 298~epochs to $0.0008$ at
  the full \num{1988}~epochs.  The coefficient estimate itself is stable
  throughout this accumulation, ranging from $-0.042$ to $-0.048$, and is
  significant at the $0.1\%$ level at every sample size tested.

\begin{figure}[H]
    \centering
    \includegraphics[width=\linewidth]{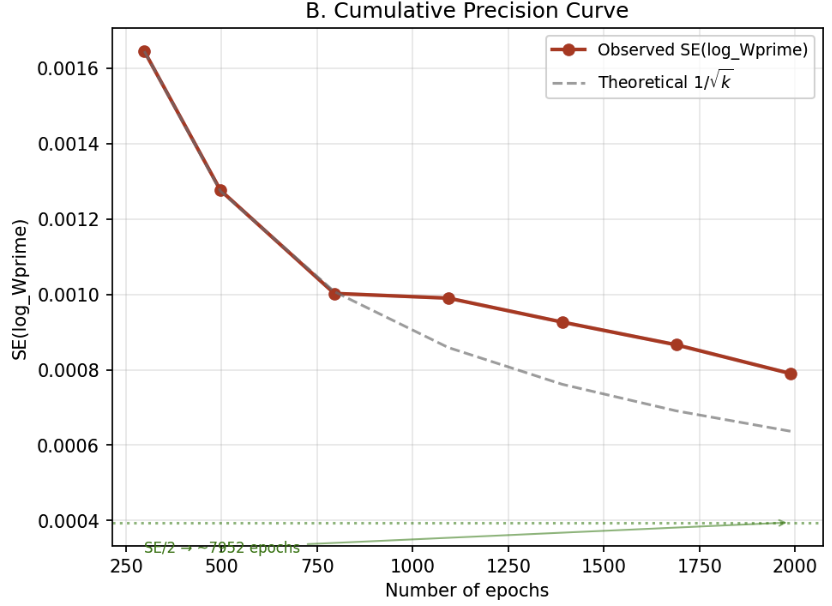}
    \caption{Cumulative precision curve for $\hat{\alpha}_1$: the
             observed SE (red) tracks the theoretical $1/\sqrt{k}$ rate (green
             dashed).}
      
    \label{fig:CPC}
  \end{figure}

  \paragraph{Rolling-window stability.}
  We re-estimate the full model on five non-overlapping temporal windows of
  approximately 400~epochs each.  The structural coefficient $\hat{\alpha}_1$
  ranges from $-0.054$ to $-0.039$ across windows
  (Figure~\ref{fig:rolling}, a spread of $8.6$ times its pooled
  standard error.  While all five estimates are negative and individually
  significant, this variation indicates non-trivial parameter drift---the delay
  gradient's effect on fees is not perfectly time-invariant.  Among the control
  variables, only $\log(\text{time since last block})$ exhibits
  window-to-window stability (range$/\text{SE} = 2.0$).

  \begin{figure}[H]
    \centering
    \includegraphics[width=\linewidth]{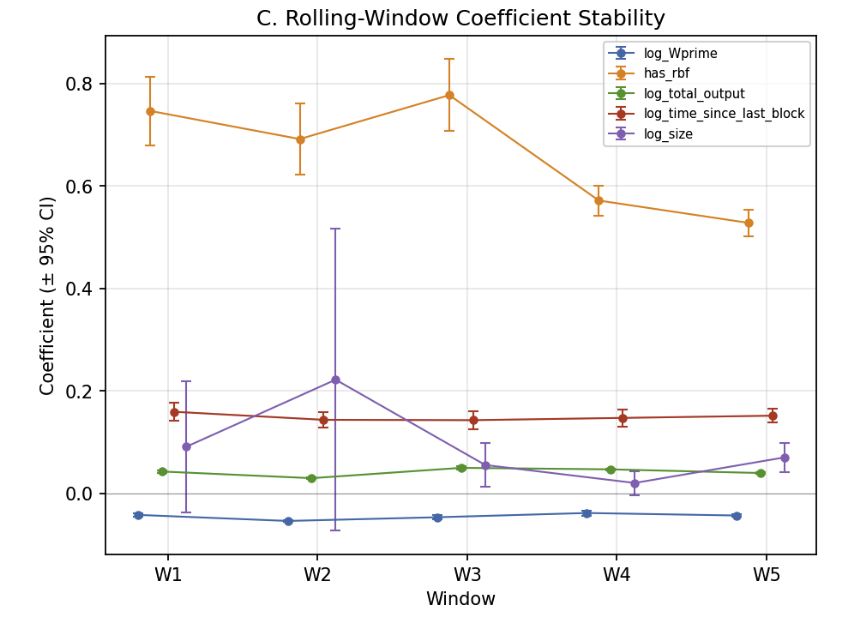}
    \caption{Rolling-window FWL coefficient estimates
             ($\pm 95\%$~CI) across five non-overlapping temporal windows.}
      
    \label{fig:rolling}
  \end{figure}

  \paragraph{Out-of-sample generalization.}
  To assess whether the structural channel generalises forward in time, we
  estimate expanding-window models via FWL that train on the first $k$ epochs
  and test on all subsequent epochs (Figure~\ref{fig:OOS}.
  We report three complementary $R^2$ measures:
  \begin{enumerate}
    \item \emph{Within-epoch, full model} (green): the training-period
      $\hat{\beta}_{\text{FWL}}$ is applied to test-epoch demeaned data,
      testing whether the structural coefficients transfer across time.
      This yields $R^2 \approx 0.24$--$0.30$ depending on the split.
    \item \emph{Within-epoch, restricted} (orange): the same procedure omitting
      $\log \hat{W}'$.  The $7$--$9$ percentage-point drop relative to the full
      model confirms that the delay gradient consistently contributes to
      within-epoch explanatory power.
    \item \emph{Strict forecasting} (red): raw $\log(\text{fee rate})$ is
      predicted using only the structural features and a training-period
      intercept, with no access to test-epoch means.  The model achieves
      $R^2 \approx 0.15$--$0.23$, confirming that the structural features alone
      carry real predictive power even without epoch fixed effects.  The dip to
      $R^2 = 0.15$ for the $80/20$ split suggests that the final portion of the
      sample may represent a distinct fee-market regime not well captured by
      earlier training data.
  \end{enumerate}

  \begin{figure}[H]
    \centering
    \includegraphics[width=\linewidth]{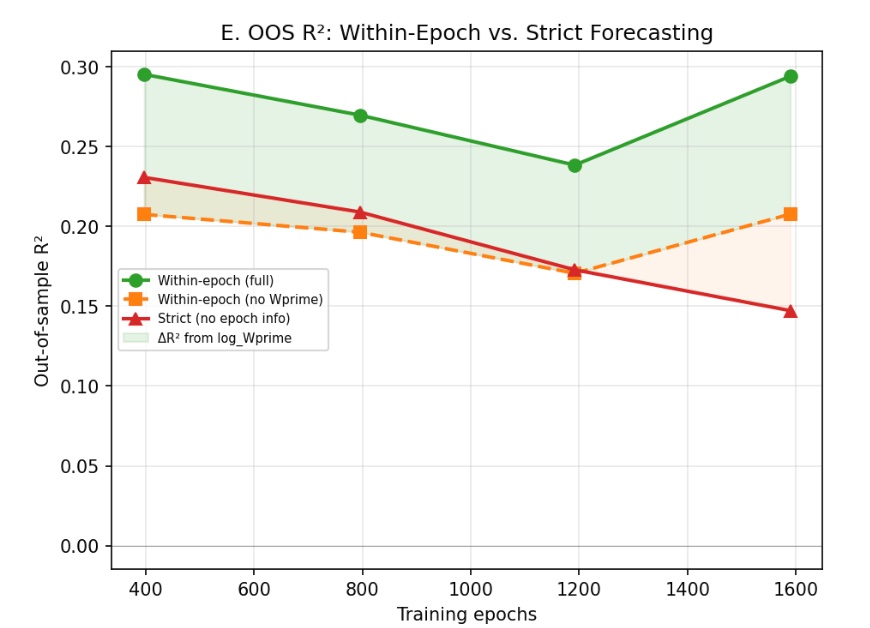}
    \caption{Out-of-sample $R^2$ learning curves under three
             evaluation protocols; the shaded wedge shows the
             $7$--$9$~pp~$\Delta R^2$ contributed by the delay gradient.}
      
    \label{fig:OOS}
  \end{figure}

  \paragraph{Variance decomposition.}
  Figure~\ref{fig:OOS2} decomposes the variance of test-period
  outcomes and strict-OOS residuals into between- and within-epoch components.
  The test outcome has $\text{ICC} = 0.195$: approximately $19\%$ of variance
  is between epochs and $81\%$ within.
  After applying the structural model without epoch information, the residual
  $\text{ICC}$ \emph{rises} to $0.292$---the model absorbs $29\%$ of
  within-epoch variance but slightly \emph{increases} between-epoch variance
  (by $21\%$).
  This asymmetry is the key finding: the structural model succeeds at explaining
  \emph{why transactions within the same time period pay different fees}---through
  the delay gradient, RBF signalling, transaction size, and other
  cross-sectional features---but it cannot predict \emph{which epochs will have
  high or low fees overall}.
  Epoch-level fee variation is driven by unobserved aggregate conditions (miner
  behaviour, macro sentiment, sudden demand shocks) that the structural features
  do not capture, and that the epoch fixed effects are designed to absorb.

 \begin{figure}[H]
    \centering
    \includegraphics[width=\linewidth]{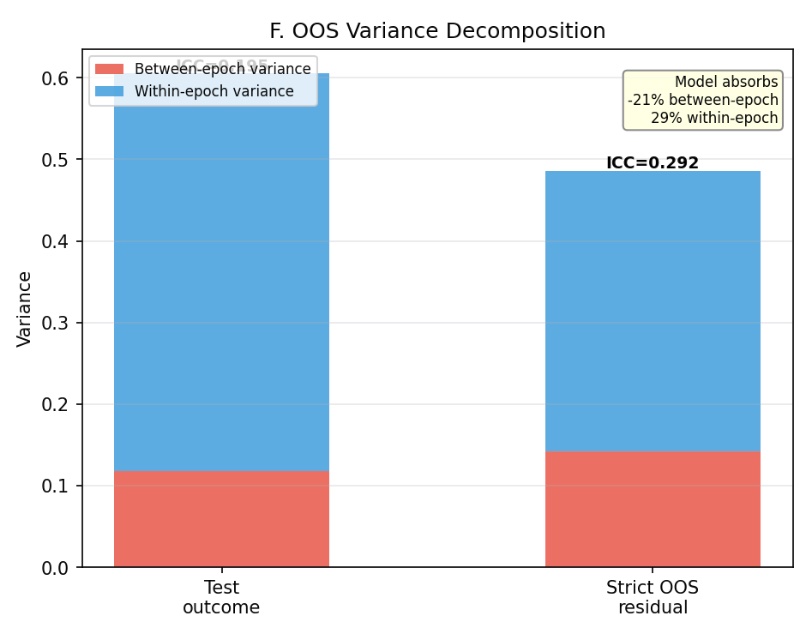}
    \caption{Variance decomposition of the $80/20$ temporal split:
             the model absorbs within-epoch variance but increases
             between-epoch variance, confirming cross-sectional rather than
             temporal explanatory power.}
      
    \label{fig:OOS2}
  \end{figure}

  \paragraph{Epoch fixed-effect persistence.}
  The autocorrelation of the estimated epoch fixed effects
  (Figure~\ref{fig:Auto} decays from $0.86$ at lag~1
  (30~minutes) to $0.10$ at lag~24 (12~hours), with a modest 24-hour seasonal
  uptick ($0.23$ at lag~48).  This pattern is consistent with slowly
  mean-reverting congestion shocks rather than permanent shifts, supporting the
  use of epoch-level fixed effects to absorb aggregate conditions.
  The strong short-run persistence also suggests that, for data collection
  purposes, spaced-out sampling (e.g.\ covering more distinct calendar days) is
  more informative than collecting consecutive hours, since adjacent epochs
  carry largely redundant aggregate information.

 \begin{figure}[H]
    \centering
    \includegraphics[width=\linewidth]{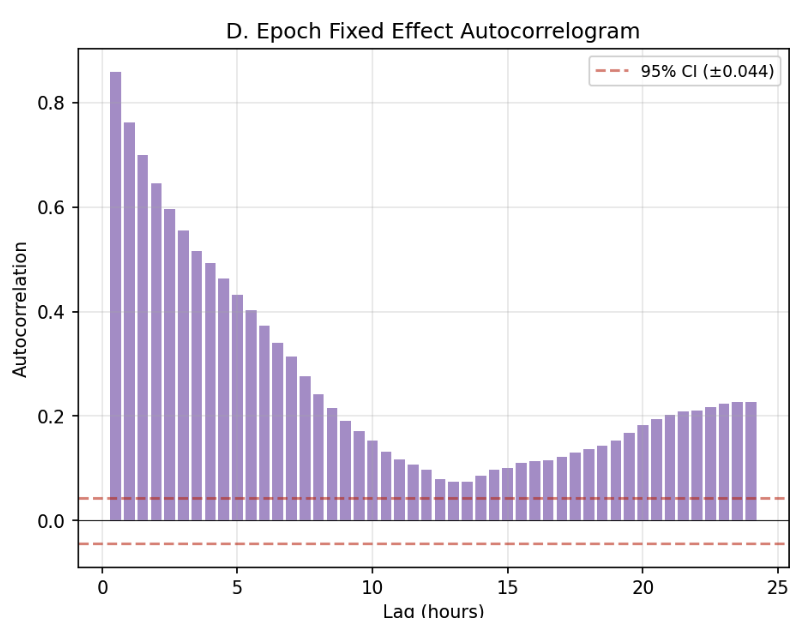}
    \caption{Autocorrelogram of epoch fixed effects.}
      
    \label{fig:Auto}
  \end{figure}

  \paragraph{Implications.}
  The structural coefficients drift across congestion regimes---$\hat{\alpha}_1$
  ranges from $-0.054$ to $-0.039$ across five temporal windows---yet the model
  generalises well out-of-sample: the VCG channel adds $7$--$9$ percentage
  points of $R^2$ within epochs even when trained on disjoint time periods.
  The variance decomposition confirms that the model's strength is
  cross-sectional: it explains within-epoch heterogeneity in fees but does not
  predict the epoch-level mean.
  Extending the sample period is therefore valuable not because the structural
  parameters converge to a fixed point, but because additional epochs provide
  independent clusters that improve precision and span a wider range of
  congestion regimes, enabling future work to model the regime-switching
  itself---for example, by interacting the structural coefficients with
  epoch-level congestion state.


\section{Conclusion}
\label{sec:conclusion}

This paper studies how Bitcoin transaction fees are formed. To address that question, we adapt a priority-pricing model of the mempool to Bitcoin's weight-based block capacity and estimate it on transaction-level data. Miners are motivated to prioritize the processing of transactions by unit fees and transactors are motivated to set their fee as a function of their relative impatience.

We introduce a novel mempool data set that enables us to observe transaction arrivals, exits, block inclusion, fee-bumping events, and high-frequency congestion states, which enables us to recover the queueing environment faced by each transaction when its fee is chosen. That information is essential for separating the technology of confirmation delay from the willingness to pay for priority, and it is largely absent from blockchain-only data.

The estimates indicate that mempool congestion governs the value of priority. Confirmation delays are shaped mainly by mempool pressure and available blockspace, and fees respond to that delay technology in the direction implied by the model. Several institutional features matter as well. RBF-tagged transactions pay substantial premia, CPFP packages are priced differently from stand-alone transactions, and contemporaneous block conditions shift fees in economically important ways. At the same time, the model is stronger in the cross section than over time. It explains why transactions in the same fee environment pay different fees better than it explains why the overall fee level moves across epochs.
That distinction is important for interpreting what the paper does and does not deliver. The estimates provide a disciplined account of private demand for transaction priority in the environment we observe, and they give the right empirical objects for counterfactual analysis. However, the short timespan of our mempool data -- approximately 90 days – limits the extent of variation in the data that is required for a confident estimate of model parameters. We intend to continue data collected from our Bitcoin ndoe to improve the preciaion of our estimates in the future.



\appendix

\section{Proof of Model Equilibrium and Discussion}
\label{app:vcg_indirect}

This appendix explains the two equilibrium properties stated in Section~3.3:
(i) fee choices implement an efficient priority allocation and
(ii) equilibrium payments coincide with the VCG payments for selling priority of service.
The key step, relative to a direct VCG mechanism, is to show how the \emph{indirect}
Bitcoin institution---\emph{fee-rate bidding} plus \emph{fee-rate-based miner selection}---implements
the VCG outcome in a large, atomistic priority-queue environment.

\subsection{Environment: fee-rate bidding as the indirect mechanism}
\label{app:env_indirect}

Fix an epoch $t$ and suppress the time subscript. As in Section~3, each transactor $i$
chooses a fee rate $r_i \ge 0$ and pays total fee $b_i = r_i \cdot \mathrm{Weight}_i$.
Expected surplus is (eq.\ (1) in the main text)
\begin{equation}
u_i = R_i - c_i W_i - b_i
    = R_i - c_i W(p_i, s) - r_i \,\mathrm{Weight}_i,
\label{eq:u_main}
\end{equation}
where $c_i$ is per-block delay cost, $s$ is the contemporaneous mempool state, and
$p_i\in(0,1)$ is the fee-rate percentile induced by the miner priority rule (eqs.\ (2)--(3)):
\begin{equation}
p_i = \frac{\mathrm{rank}(r_i)}{N}, 
\qquad
W_i = W(p_i,s),
\label{eq:pW_main}
\end{equation}
with higher $p$ meaning higher priority (shorter expected delay).

\paragraph{Miner behavior.}
Under competitive mining and absent strategic withholding, miners myopically maximize fee
revenue per unit of blockspace and therefore select transactions in descending order of fee rate
until the weight constraint binds (this is the Bitcoin adaptation of the priority-queue framework
of \citet{Huberman2021}).

\paragraph{Atomistic transactors (price-taking in the priority market).}
A maintained assumption for both the model and the empirical strategy is that each transaction is
small relative to the mempool and takes as given the mapping from fee rates to expected delay in the
current state $s$. Formally, a single transaction treats the function $r \mapsto p(r)$ and
$p \mapsto W(p,s)$ as fixed at the time of fee choice. This is the sense in which fee-rate bidding
can be analyzed as choosing a point on an equilibrium ``priority price schedule'' rather than as an
$\epsilon$-undercutting game among a finite set of players.

\subsection{VCG is an appropriate formalization of the Bitcoin fee market}
\label{app:why_vcg}

Section~3 is a congested-service model: higher priority means faster confirmation, and priority is
scarce because block capacity is limited. The only strategic interaction is a \emph{congestion externality}:
when transaction $i$ increases its priority, it displaces lower-priority transactions and increases their
expected confirmation delays.

The Vickrey--Clarke--Groves (VCG) mechanism \citep{vickrey1961counterspeculation,clarke1971multipart,groves1973incentives}
is precisely the canonical way to (i) implement an efficient allocation of a scarce resource under private
information and (ii) charge each agent the external cost it imposes on others. In the present setting,
that external cost is the \emph{expected delay cost} imposed on lower-priority transactions.

A central result in \citet{Huberman2021} is that, in the priority-queue fee market, the decentralized
equilibrium payment schedule coincides with these VCG externality payments. Section~3.3 adopts this interpretation:
fees are prices for priority, equal to expected delay externalities on lower-priority transactors.








\subsection{Equilibrium of the indirect mechanism: fee-rate bidding}
\label{app:equilibrium_indirect}

We now show how the \emph{indirect} mechanism (users choose fee rates; miners sort by fee rates)
induces a payment schedule that coincides with VCG.

\paragraph{Step 1: reduce fee-rate choice to choosing a priority percentile.}
Let $G(\cdot)$ denote the equilibrium cross-sectional distribution of fee rates in the relevant set used
to define percentile rank in (2). For a price-taking transaction, choosing a fee rate $r$ is equivalent to
choosing a percentile $p = G(r)$; if $G$ is strictly increasing on its support, it has an inverse
$r(p) \equiv G^{-1}(p)$.

Thus, conditional on the prevailing schedule $r(\cdot)$ and state $s$, a transaction with characteristics
$(R_i,c_i,\mathrm{Weight}_i)$ solves
\begin{equation}
\max_{p \in [0,1]}
\Big\{ R_i - c_i W(p,s) - \mathrm{Weight}_i \, r(p) \Big\}.
\label{eq:user_problem_p}
\end{equation}

\paragraph{Notation for slopes.}
To avoid sign confusion, define the \emph{positive} local slope of the delay schedule
\begin{equation}
D(p,s) \equiv -\frac{\partial W(p,s)}{\partial p} > 0.
\label{eq:D_def}
\end{equation}
(We can interpret the $W_p(\cdot)$ term in eq.\ (4) of the main text as this positive slope.)

\paragraph{Step 2: monotone sorting in equilibrium.}
Because higher $c_i$ increases the marginal benefit of higher priority, the model has a single-crossing
property in $(c,p)$.

\begin{lemma}[Single-crossing $\Rightarrow$ monotone equilibrium priorities]
\label{lem:single_crossing}
Suppose $r(p)$ is weakly increasing and $W(p,s)$ strictly decreasing. Then any equilibrium of
\eqref{eq:user_problem_p} in which types separate is monotone: higher $c_i$ choose weakly higher $p_i$.
\end{lemma}

\begin{proof}
Let $V(p)\equiv -W(p,s)$ so $V$ is strictly increasing in $p$.
Type $c$ chooses $p$ to maximize $R + cV(p) - \mathrm{Weight}\,r(p)$.
The cross-partial is $\partial^2/\partial c\,\partial p \,[cV(p)] = V'(p) >0$,
so preferences satisfy increasing differences in $(c,p)$. By standard monotone comparative statics,
optimal $p$ is nondecreasing in $c$.
\end{proof}

Lemma~\ref{lem:single_crossing} delivers the first part of the equilibrium claim in Section~3.3:
higher delay-cost transactors select higher priority (and thus shorter delay).

\paragraph{Step 3: the equilibrium marginal condition pins down the payment rule.}
Assume an interior solution for $p$ and differentiability of $r(\cdot)$.

The first-order condition for \eqref{eq:user_problem_p} is
\begin{equation}
c_i D(p_i,s) = \mathrm{Weight}_i \, r'(p_i).
\label{eq:foc}
\end{equation}

In a separating (monotone) equilibrium, there is a well-defined delay-cost type $c(p)$ at each percentile $p$
(up to measure-zero issues). Taking \eqref{eq:foc} at the type that chooses percentile $p$ yields the
\emph{equilibrium marginal fee-rate schedule}:
\begin{equation}
r'(p) = \frac{c(p)}{\overline{\mathrm{Weight}}(p)}\, D(p,s),
\label{eq:rprime_general}
\end{equation}
where $\overline{\mathrm{Weight}}(p)$ denotes the weight of the marginal transaction at percentile $p$.
In the equal-weight benchmark of \citet{Huberman2021}, or under the approximation that marginal weights
do not vary sharply over $p$ within an epoch, $\overline{\mathrm{Weight}}(p)$ can be treated as constant.
In that case the equilibrium fee schedule is pinned down (up to an additive constant) by integrating
\eqref{eq:rprime_general}.

For the Bitcoin adaptation, it is convenient to bundle the mapping from percentile units to ``how much blockspace is
displaced when priority rises'' into a normalization constant $\kappa(s)$ (this is the $\kappa_t$ in eq.\ (4) of
the main text). Writing this compactly, the implied \emph{marginal} total fee for a transaction of weight
$\mathrm{Weight}_i$ at percentile $p$ is
\begin{equation}
\frac{d b(p)}{dp}
= \mathrm{Weight}_i\, r'(p)
= \kappa(s)\, c(p)\, D(p,s)\,\mathrm{Weight}_i.
\label{eq:dbdp}
\end{equation}

\subsection{VCG payments for priority and equivalence to the bidding equilibrium}
\label{app:vcg_equiv}

We now compute the VCG payment for ``selling priority of service'' and show it matches the fee schedule implied by
fee-rate bidding.

\paragraph{VCG payment as delay externality (finite-$n$ intuition).}
Consider (for intuition) a finite set of $n$ equal-size transactions ordered by priority (highest first),
with expected delays $W_1 \le W_2 \le \dots \le W_n$ and delay costs $c_1,\dots,c_n$.
Removing transaction $m$ shifts every lower-priority transaction $j>m$ up by one rank, reducing its expected delay
from $W_j$ to $W_{j-1}$. The externality imposed by transaction $m$ on lower-priority transactions is therefore
\begin{equation}
b_m^{VCG}
= \sum_{j=m+1}^{n} c_j (W_j - W_{j-1}).
\label{eq:vcg_discrete}
\end{equation}
This is the VCG payment: the increase in others' total delay cost caused by $m$'s presence.

\paragraph{Continuous approximation.}
In a large market, index transactions by percentile $p\in[0,1]$ and write $c(p)$ for the delay cost of the
transaction at percentile $p$. A small increase $dp$ in priority shifts down an infinitesimal mass of lower-priority
transactions and increases their expected delay by approximately $D(p,s)\,dp$ blocks. Valuing that marginal delay at
$c(p)$ yields the \emph{marginal} externality price:
\[
d b^{VCG}(p)
= \kappa(s)\, c(p)\, D(p,s)\,\mathrm{Weight}_i \, dp.
\]
Integrating from $0$ (lowest priority) to $p$ gives the VCG payment schedule:
\begin{equation}
b^{VCG}(p)
= \kappa(s)\,\mathrm{Weight}_i \int_{0}^{p} c(q)\, D(q,s)\, dq.
\label{eq:vcg_cont}
\end{equation}

\begin{proposition}[Fee-rate bidding equilibrium payments coincide with VCG]
\label{prop:equiv_vcg}
Assuming transaction delay is monotone in fee rank and the atomistic (price-taking) assumption, any differentiable monotone
equilibrium of fee-rate bidding satisfies
\[
\frac{db(p)}{dp} = \kappa(s)\,c(p)\,D(p,s)\,\mathrm{Weight}_i,
\]
and hence its implied payment schedule coincides with the VCG payment schedule \eqref{eq:vcg_cont}.
\end{proposition}

\begin{proof}
From the fee-rate bidding equilibrium, the marginal fee schedule must satisfy \eqref{eq:dbdp}. Integrating that
marginal condition from $0$ to $p$ and normalizing so that the lowest-priority transaction pays zero yields
\eqref{eq:vcg_cont}. But \eqref{eq:vcg_cont} is exactly the VCG externality payment for priority, as derived from the
delay-externality logic (finite-$n$ formula \eqref{eq:vcg_discrete} and its continuous approximation).
Therefore the equilibrium payments coincide with VCG.
\end{proof}

\subsection{Connection to the approximation in eqs.\ (4)--(5)}
\label{app:connect_eq4_eq5}

Equation \eqref{eq:vcg_cont} is an exact (continuous) VCG payment formula. For estimation, Section~3.3 uses the local
approximation (eq.\ (4) in the main text). One way to see the link is via a mean-value approximation of the integral:
there exists $\tilde p\in(0,p)$ such that
\[
\int_{0}^{p} c(q) D(q,s)\,dq
= c(\tilde p) D(\tilde p,s)\cdot p.
\]
Bundling $p$ and any percentile-to-displacement scaling into $\kappa(s)$ yields the approximation
\begin{equation}
b_i \approx \kappa(s)\, c_i \, D(p_i,s)\,\mathrm{Weight}_i,
\label{eq:eq4_like}
\end{equation}
which matches eq.\ (4) (with $D$ corresponding to the positive local slope denoted by $W_p$ in the main text).
Taking logs yields the additive decomposition in eq.\ (5).









\section{Identification and inference}
\label{app:identification}

This appendix formalizes identification of the fee decomposition in eq.~(5) and the two-stage estimator in eqs.~(9)--(13). The key empirical challenge is that the delay-technology term $\log W_p(p_{it},s_t)$ in eq.~(5) is not observed. Our strategy is to (i) estimate the priority--delay mapping from delay outcomes and recover an
observation-level estimate of the local slope of the delay schedule, and then (ii) use that estimated slope as a control in the fee equation to identify the mapping from the impatience proxy into the latent delay cost.

\paragraph{Targets and estimands.}
Define $x_{it}$ as the vector of first-stage predictors in eq.~(9):

\[x_{it}\equiv\Big(p_{it},\hat\rho_t,\mathrm{Blockspace}_t,\log\mathrm{Weight}_{it}, \mathrm{RBF}_{it},\mathrm{CPFP}_{it},\mathrm{Type}_{it}\Big),\]

and let $g_0(x)\equiv \mathbb{E}[W_{it}\mid x_{it}=x]$ denote the population delay technology. Given an increment $\Delta$ (eq.~(11)), define the associated population finite-difference slope (the object our estimator targets with $\Delta$ fixed):

\begin{equation}
W'_{0,it}(\Delta)\;\equiv\;\frac{g_0(p_{it}+\Delta,z_{it})-g_0(p_{it} \Delta,z_{it})}{2\Delta},
\label{eq:Wprime_theta}
\end{equation}

where $z_{it}$ stacks the non-$p$ elements of $x_{it}$ and $p_{it}\pm\Delta$ are interpreted with the same clipping to $[0,1]$ used in eq.~(11).

The second-stage estimand is the best linear/spline approximation to the preference component $\log c_{it}$ in eq.~(5) using the impatience proxy $\iota_{it}$ (eq.~(7)), conditional on the delay technology proxy $\log W'_{0,it}(\Delta)$, observed controls, and epoch fixed effects:

\begin{equation}
\log b_{it} = \alpha_0 + B(\iota_{it})'\delta + \alpha_1 \log W'_{0,it}(\Delta) + X_{it}'\beta + \eta_t
+ \varepsilon_{it}.
\label{eq:secondstage_population}
\end{equation}

Here $B(\iota_{it})$ is the cubic B-spline basis in eq.~(13), and $X_{it}$ is the control set in eq.(12) (including $\log\mathrm{Weight}_{it}$ and type/wallet indicators).

\subsection*{Maintained assumptions}

\paragraph{A1 (Priority sufficiency).}
Within an epoch, conditional on the covariates in $x_{it}$, transaction priority is summarized by the fee-rate percentile $p_{it}$ defined in eq.~(8), in the sense that $\mathbb{E}[W_{it}\mid r_{it},x_{it}] = \mathbb{E}[W_{it}\mid x_{it}]$. Intuitively, once we condition on the percentile index and the RBF/CPFP/package-related controls, the remaining level of the fee rate does not shift expected delay except through its rank.

\paragraph{A2 (Delay technology exogeneity).} The first-stage error $u_{it}$ in eq.~(9) satisfies $\mathbb{E}[u_{it}\mid x_{it}]=0$.
Equivalently, conditional on priority and predetermined state, residual delay variation is mean-zero block-arrival randomness and other orthogonal shocks.

\paragraph{A3 (Smoothness and positivity).}
For each epoch state, $g_0(p,z)$ is continuous in $(p,z)$ and locally smooth in $p$ on the support used in estimation, and $W'_{0,it}(\Delta)>0$ with probability one on the trimmed sample so that
$\log W'_{0,it}(\Delta)$ is well defined.

\paragraph{A4 (Atomistic transactions).}
Each transaction is small relative to the mempool within its epoch. Consequently, a single transaction’s unobserved fee component does not shift the technology $g_0(\cdot)$ or its local slope, and the realized
pair $(W_{it},x_{it})$ can be treated as generated by the environment rather than by the transaction’s idiosyncratic second-stage shock.

\paragraph{A5 (Preference proxy restriction).}
The impatience proxy $\iota_{it}$ shifts delay costs but is conditionally mean independent of the second-stage disturbance:

\[\mathbb{E}[\varepsilon_{it} \mid B(\iota_{it}), \log W'_{0,it}(\Delta), X_{it}, \eta_t]=0.\]

This is the formal content of interpreting $\iota_{it}$ as a proxy for the latent delay cost component in eq.(5), after controlling for the delay technology and observables.

\paragraph{A6 (First-stage consistency with cross-fitting).}
Let $\widehat g^{(-k)}(\cdot)$ be the cross-fitted first-stage estimator defined in eq.~(10), and let $\widehat W'_{it}$ be the finite-difference slope in eq.~(11). Then, as the sample grows,

\[ \mathbb{E}\big[\big(\widehat g^{(-k)}(x_{it})-g_0(x_{it})\big)^2\big]\to 0\quad \text{and} \quad
\mathbb{E}\big[\big(\widehat W'_{it}-W'_{0,it}(\Delta)\big)^2\big]\to 0,\]

with the cross-fitting construction ensuring that first-stage estimation error is not mechanically fitted to the held-out observation's outcome $W_{it}$.\footnote{This is the standard motivation for cross-fitting and for treating second-stage regressions with generated ML regressors; see, e.g., \citep{pagan1984generated} and the cross-fitting discussion in \citep{chernozhukov2018dml}.}

\subsection{Identification}

\paragraph{Proposition 1 (Population identification).} Under A1--A5 and a standard full-rank condition for the regressors in \eqref{eq:secondstage_population}, the parameter vector $\theta_0\equiv(\alpha_0,\delta,\alpha_1,\beta)$ is uniquely identified as the minimizer of the population least squares criterion:

\[ \theta_0 = \arg\min_{\theta}\; \mathbb{E}\left[ \left(\log b_{it}-\alpha_0 -B(\iota_{it})'\delta -
\alpha_1 \log W'_{0,it}(\Delta) - X_{it}'\beta -\eta_t\right)^2\right].\]

Moreover, $\theta_0$ identifies the mapping from the impatience proxy into the latent delay-cost component in eq.~(5) up to the approximation error inherent in (i) the local pricing approximation in eq.(4) and
(ii) the finite-difference slope definition in \eqref{eq:Wprime_theta} for fixed $\Delta$.

\textbf{Discussion.} Proposition 1 formalizes why the two stages are not redundant. Stage 1 identifies the technology object $W'_{0,it}(\Delta)$ from delay outcomes (eqs.~(9)--(11)); stage 2 then identifies the
preference mapping $B(\iota_{it})'\delta$ in the fee equation conditional on this technology control. Absent stage 1, variation in fees induced by changes in the steepness of the priority--delay schedule would be absorbed by the impatience mapping, confounding preferences and technology.

\subsection*{Consistency of the two-stage estimator}

Let $\widehat\theta$ denote the OLS estimator from the feasible second-stage equation (12) that replaces
$\log W'_{0,it}(\Delta)$ with $\log \widehat W'_{it}$.

\paragraph{Proposition 2 (Consistency of the feasible two-stage estimator).} Under A1--A6 and standard regularity conditions for OLS with epoch fixed effects and epoch-level clustering,

\[ \widehat\theta \xrightarrow{p} \theta_0. \]
That is, using $\log \widehat W'_{it}$ in place of the infeasible $\log W'_{0,it}(\Delta)$ does not change
the probability limit of the second-stage estimator.

\textbf{Informal argument.} Define the infeasible OLS estimator $\widetilde\theta$ that uses the true regressor $\log W'_{0,it}(\Delta)$. By A5, $\widetilde\theta$ is consistent for $\theta_0$ by standard partialling-out arguments (with $\eta_t$). By A6 and the continuous mapping theorem, $\log \widehat W'_{it}-\log W'_{0,it}(\Delta)\to 0$ in mean square on the estimation sample (after trimming where needed). Cross-fitting ensures the first-stage error is not constructed from the same observation’s realized delay in a way that would induce mechanical correlation with the second-stage disturbance. Under these conditions, the feasible and infeasible normal equations converge to the same population moments, implying $\widehat\theta-\widetilde\theta=o_p(1)$ and hence $\widehat\theta \xrightarrow{p} \theta_0$.

\paragraph{Why this avoids ``fee rate on both sides'' bias.} The second stage does not include the fee rate $r_{it}=b_{it}/\mathrm{Weight}_{it}$ (or the percentile $p_{it}$) as a regressor. Including $r_{it}$ would create a mechanical relationship with $\log b_{it}$ because $\log r_{it}=\log b_{it}-\log\mathrm{Weight}_{it}$. Instead, the second stage controls for $\log \widehat W'_{it}$, which is constructed from the \emph{delay} equation and summarizes the technological constraint faced by a price-taking transactor (A4). This is precisely the control implied by eq.~(5).

\subsection{Inference}

We report standard errors clustered at the epoch level to allow arbitrary within-epoch dependence. To incorporate first-stage estimation uncertainty, we additionally report an epoch-block bootstrap that resamples epochs and re-estimates both stages within each resample.


\section{Additional robustness checks}
\label{app:spline}

\subsection{Spline specification}
\label{app:spline:spec}

As a first robustness check, we re-estimate the Stage~2 regression
replacing the linear specification in equation~\eqref{eq:intensive} with a
monotone I-spline regression in the impatience proxy
$\iota_{it} = 1/(\text{min respend blocks}_{it} + \varepsilon)$.
The spline uses seven basis functions (cubic degree, five quantile-based knot
locations), and the basis coefficients are constrained to be weakly
non-negative so that the estimated impatience mapping is weakly increasing.
All other components of the specification are unchanged: the regressors still
include $\log \widetilde{W}'_{it}$, the full set of transaction controls,
block-state controls, and epoch fixed effects, with standard errors clustered
at the epoch level.

\paragraph{Main result.}
The structural coefficient on the delay gradient is essentially unchanged in
the spline specification:
\[
  \widehat{\alpha}_1 = -0.0455
  \qquad
  (\text{clustered SE } = 0.0008,\; t = -57.3).
\]

\begin{figure}[H]
    \centering
    \includegraphics[width=\linewidth]{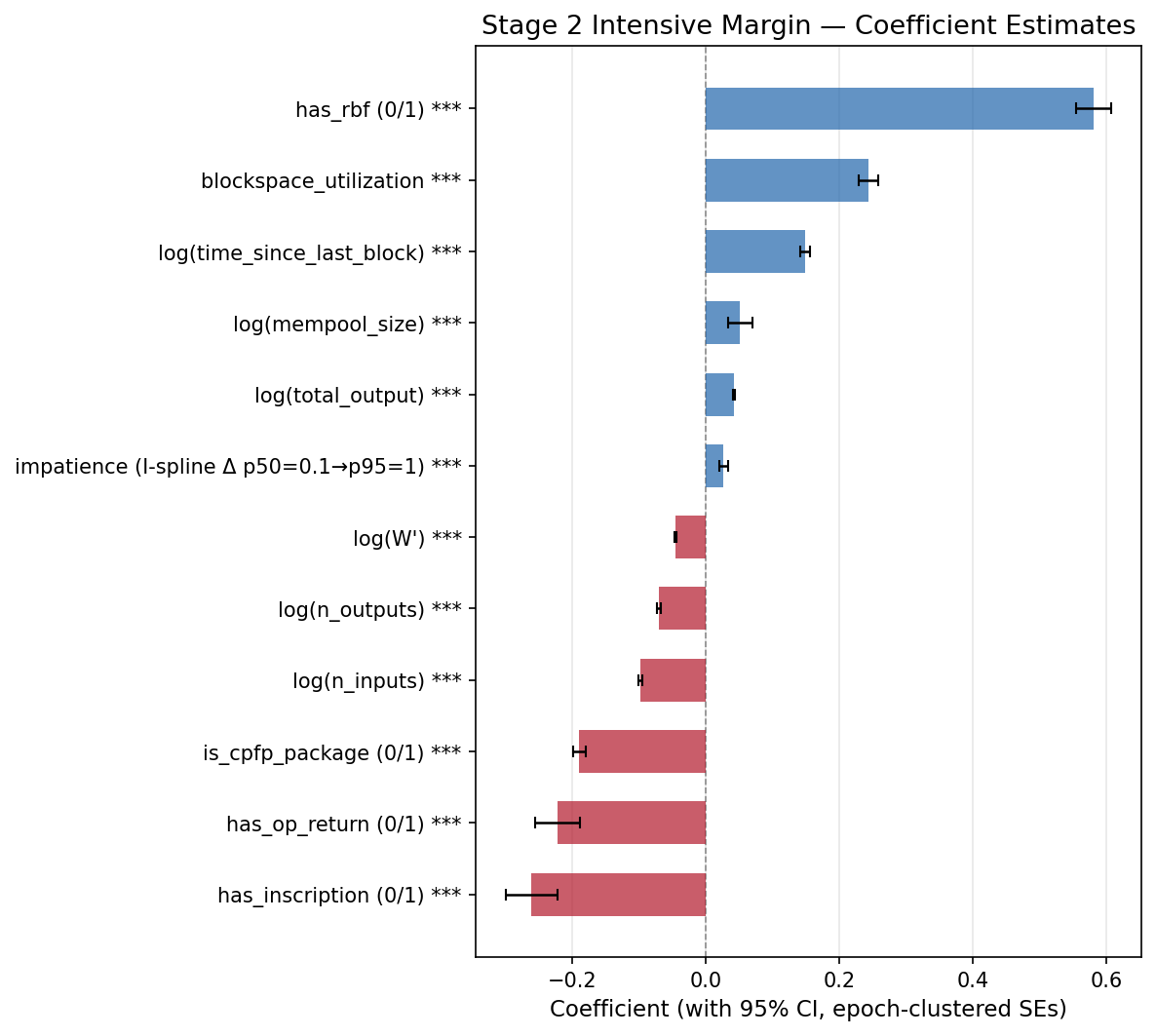}
    \caption{Second-stage coefficient estimates with $95\%$ confidence
             intervals (with spline).  Blue bars indicate
             positive effects on log fee rate; red bars indicate negative
             effects. Impatience has little to no effect. }
    \label{fig:spline}
  \end{figure}

This is nearly identical to the baseline linear specification, indicating that
the core VCG channel is not an artefact of imposing a linear impatience term or
omitting impatience entirely.

\paragraph{Spline fit and aggregate effect.}
The spline specification attains in-sample $R^2 = 0.4382$ with a Duan
smearing factor of $1.3412$, a modest improvement over the baseline fit.
To summarise the impatience effect in a way comparable to the blockspace spline,
we aggregate the I-spline basis via the delta method and report the implied
change in $\log(\text{fee rate})$ when impatience moves from its sample median
to its 95th percentile.  This yields
\[
  \Delta \log(\text{fee rate})
  \big|_{\iota:\,p50=0.1 \rightarrow p95=1}
  = 0.0270
  \qquad
  (\text{SE } = 0.0032,\; p < 0.001).
\]
Thus, higher impatience is associated with a statistically significant but
economically modest increase in fee rates, conditional on the delay gradient
and the full set of controls.

\begin{figure}[H]
    \centering
    \includegraphics[width=\linewidth]{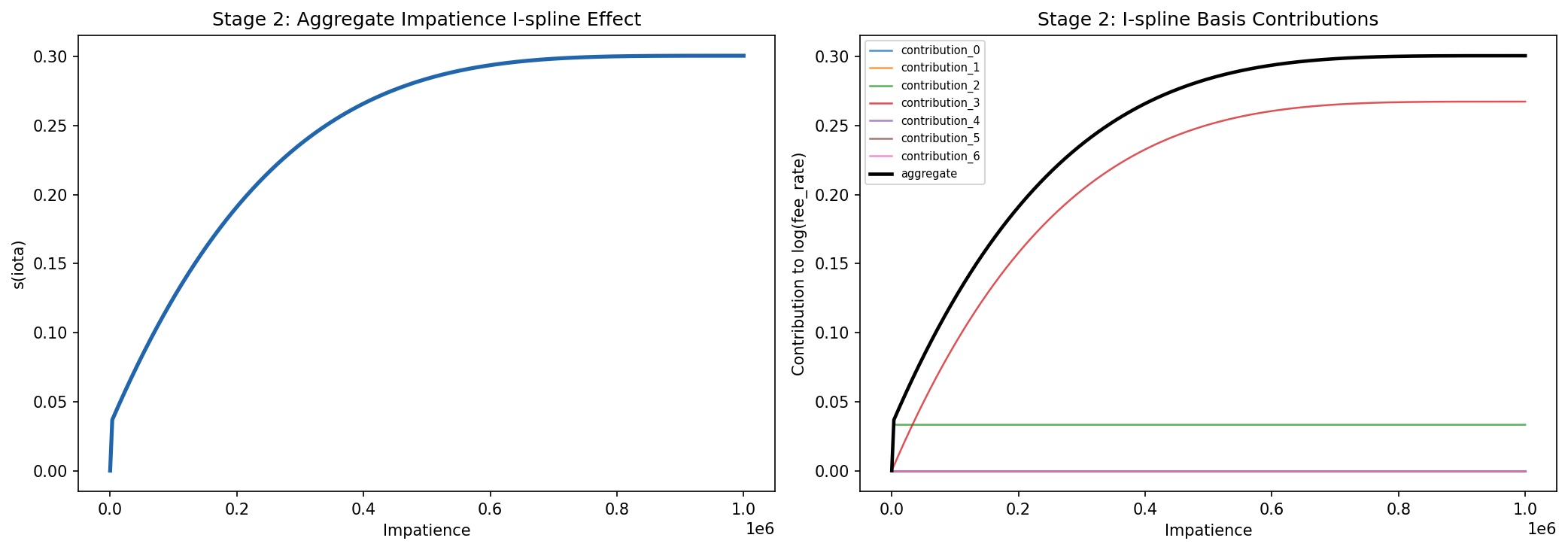}
    \caption{Aggregate impatience effect. Left graph is the sum of each of the 7 spline contributions as seen in the right graph. }
    \label{fig:spline}
  \end{figure}

\paragraph{Interpretation.}
The estimated spline is monotone increasing, as imposed by the non-negativity
constraints on the I-spline coefficients, and the fitted shape is concentrated
in the interior of the impatience distribution rather than at the extreme tail.
Most importantly, allowing a flexible non-parametric impatience mapping does
not materially alter the estimate on $\log \widetilde{W}'_{it}$.
The main structural conclusion therefore survives this richer specification:
cross-transaction variation in fees is still primarily explained by the delay
gradient channel, while impatience contributes a smaller second-order
adjustment.

\subsection{Alternative first-stage learner: monotone XGBoost}
\label{app:spline:xgb}

As a second robustness check, we replace the Random Forest first stage with a
monotone gradient-boosted tree learner (XGBoost) that imposes the required
weakly decreasing restriction in priority $p_{it}$ by construction while still
allowing flexible interactions with congestion variables.  Using the resulting
delay-gradient estimates $\widehat W'^{\,\mathrm{MGB}}_{it}$, we re-estimate the
same Stage~2 spline specification.

The full-sample monotone-XGBoost run uses the same \num{11110277}
observations and \num{1988} epochs as the Random-Forest benchmark, making the
comparison directly like-for-like.  In Stage~1, the monotone XGBoost attains a
held-out test fit of $R^2 = 0.5653$ (RMSE $= 0.8579$), somewhat below the
Random Forest benchmark.  Nonetheless, the Stage~2 spline specification fitted
on the resulting $\widehat W'^{\,\mathrm{MGB}}_{it}$ achieves an in-sample
$R^2 = 0.6856$ with a smearing factor of $1.1687$.

\paragraph{Structural coefficient comparison.}
The structural coefficient on the delay gradient remains negative and highly
statistically significant under the monotone-XGBoost first stage,
\[
  \widehat{\alpha}_1^{\,\mathrm{MGB}} = -0.0861
  \qquad
  (\text{clustered SE } = 0.0005,\; t = -185.0).
\]
Relative to the Random-Forest spline specification
($\widehat{\alpha}_1 = -0.0455$), the coefficient is substantially more
negative in magnitude, though its sign and high precision are unchanged.  The
qualitative conclusion of the model therefore survives: the delay-gradient
channel remains the dominant structural determinant of within-epoch fee
dispersion, but the exact magnitude of the effect is sensitive to the
first-stage learner used to recover $\widehat W'_{it}$.

\paragraph{Impatience spline comparison.}
Aggregating the I-spline basis via the same delta-method contrast as above
gives
\[
  \Delta \log(\text{fee rate})
  \big|_{\iota:\,p50=0.1 \rightarrow p95=1}
  = 0.0178
  \qquad
  (\text{SE } = 0.0023,\; p < 0.001).
\]
This is smaller than the corresponding Random-Forest estimate
($0.0270$ with clustered SE $0.0032$), but it remains positive and
statistically significant.  Hence the monotone-XGBoost first stage preserves
the qualitative shape of the estimated impatience mapping while attenuating its
aggregate magnitude.

\paragraph{Interpretation.}
Taken together, the two first-stage learners deliver the same broad qualitative
message: the impatience spline is increasing, the impatience effect is positive
but modest, and the delay-gradient coefficient is negative and precisely
estimated.  The main difference is quantitative rather than qualitative.
Relative to the Random-Forest first stage, the monotone-XGBoost first stage
produces a steeper estimated structural response to $\log \widetilde{W}'_{it}$
and a somewhat smaller aggregate impatience effect.  We therefore interpret the
core economic mechanism as robust, while treating the exact coefficient
magnitude as first-stage-specification-sensitive.


\begin{thebibliography}{9}
\providecommand{\natexlab}[1]{#1}
\providecommand{\url}[1]{\texttt{#1}}
\expandafter\ifx\csname urlstyle\endcsname\relax
  \providecommand{\doi}[1]{doi: #1}\else
  \providecommand{\doi}{doi: \begingroup \urlstyle{rm}\Url}\fi

\bibitem[Chernozhukov et~al.(2018)Chernozhukov, Chetverikov, Demirer, Duflo,
  Hansen, Newey, and Robins]{chernozhukov2018dml}
Victor Chernozhukov, Denis Chetverikov, Mert Demirer, Esther Duflo, Christian
  Hansen, Whitney Newey, and James Robins.
\newblock Double/debiased machine learning for treatment and structural
  parameters.
\newblock \emph{The Econometrics Journal}, 21\penalty0 (1):\penalty0 C1--C68,
  2018.
\newblock \doi{10.1111/ectj.12097}.

\bibitem[Clarke(1971)]{clarke1971multipart}
Edward~H. Clarke.
\newblock Multipart pricing of public goods.
\newblock \emph{Public Choice}, 11\penalty0 (1):\penalty0 17--33, 1971.
\newblock \doi{10.1007/BF01726210}.

\bibitem[Easley et~al.(2019)Easley, O'Hara, and Basu]{Easley2019}
David Easley, Maureen O'Hara, and Soumya Basu.
\newblock From mining to markets: The evolution of bitcoin transaction fees.
\newblock \emph{Journal of Financial Economics}, 134\penalty0 (1):\penalty0
  91--109, 2019.
\newblock ISSN 0304-405X.
\newblock \doi{https://doi.org/10.1016/j.jfineco.2019.03.004}.
\newblock URL
  \url{https://www.sciencedirect.com/science/article/pii/S0304405X19300583}.

\bibitem[Groves(1973)]{groves1973incentives}
Theodore Groves.
\newblock Incentives in teams.
\newblock \emph{Econometrica}, 41\penalty0 (4):\penalty0 617--631, 1973.
\newblock \doi{10.2307/1914085}.

\bibitem[Huberman et~al.(2021)Huberman, Leshno, and Moallemi]{Huberman2021}
Gur Huberman, Jacob~D Leshno, and Ciamac Moallemi.
\newblock {Monopoly without a Monopolist: An Economic Analysis of the Bitcoin
  Payment System}.
\newblock \emph{The Review of Economic Studies}, 88\penalty0 (6):\penalty0
  3011--3040, 03 2021.
\newblock \doi{10.1093/restud/rdab014}.
\newblock URL \url{https://doi.org/10.1093/restud/rdab014}.

\bibitem[Lehar and Parlour(2020)]{lehar2020miner}
Alfred Lehar and Christine~A. Parlour.
\newblock Miner collusion and the bitcoin protocol.
\newblock Technical report, SSRN, March 22 2020.
\newblock URL \url{https://ssrn.com/abstract=3559894}.

\bibitem[M\"oser and B\"ohme(2015)]{Bohme2015}
Malte M\"oser and Rainer B\"ohme.
\newblock Trends, tips, tools: A longitudinal study of bitcoin transaction
  fees.
\newblock In M.~Brenner, N.~Christin, B.~Johnson, and K.~Rohloff, editors,
  \emph{Financial Cryptography and Data Security, 2nd Workshop on BITCOIN
  Research}, volume 8976 of \emph{Lecture Notes in Computer Science}, pages
  19--33. Springer, 2015.

\bibitem[Pagan(1984)]{pagan1984generated}
Adrian Pagan.
\newblock Econometric issues in the analysis of regressions with generated
  regressors.
\newblock \emph{International Economic Review}, 25\penalty0 (1):\penalty0
  221--247, 1984.
\newblock \doi{10.2307/2526417}.

\bibitem[Vickrey(1961)]{vickrey1961counterspeculation}
William Vickrey.
\newblock Counterspeculation, auctions, and competitive sealed tenders.
\newblock \emph{The Journal of Finance}, 16\penalty0 (1):\penalty0 8--37, 1961.
\newblock \doi{10.1111/j.1540-6261.1961.tb02789.x}.

\end{thebibliography}
\end{document}